\documentclass[12pt,reqno]{amsart}
\usepackage{dsfont, amssymb,amsmath,amscd,latexsym, amsthm, amsxtra,amsfonts,enumerate}
\usepackage[all]{xy}
\usepackage[active]{srcltx}
\usepackage{mathrsfs}
\usepackage[all]{xy}
\usepackage[active]{srcltx}
\usepackage{graphicx}
\usepackage{bm}
\usepackage{bbm}
\usepackage{graphicx}
\usepackage{tabularx}
\usepackage{caption}
\usepackage{color}
\usepackage{algorithmicx,algorithm}
\usepackage{longtable}
\usepackage[round]{natbib}
\usepackage{ulem}
\usepackage{verbatim}
\usepackage{appendix}
\usepackage{subcaption}
\usepackage[pdfstartview=FitH, bookmarksnumbered=true,bookmarksopen=true, colorlinks=true, pdfborder=001, citecolor=blue, linkcolor=blue,urlcolor=blue]{hyperref}
\usepackage{CJK}
\newtheorem{theorem}{Theorem}[section]
\newtheorem{lemma}[theorem]{Lemma}

\newtheorem{proposition}[theorem]{Proposition}

\newtheorem{remark}[theorem]{Remark}

\begin{document}
\makeatletter
\def\@setauthors{%
\begingroup
\def\thanks{\protect\thanks@warning}%
\trivlist \centering\footnotesize \@topsep30\p@\relax
\advance\@topsep by -\baselineskip
\item\relax
\author@andify\authors
\def\\{\protect\linebreak}%
{\authors}%
\ifx\@empty\contribs \else ,\penalty-3 \space \@setcontribs
\@closetoccontribs \fi
\endtrivlist
\endgroup } \makeatother
 \baselineskip 19pt
\title[{{\tiny Despite Absolute Information Advantages, All Investors Incur Welfare Loss }}]
 {{\tiny Despite Absolute Information Advantages, All Investors Incur Welfare Loss }}
 \vskip 10pt\noindent
\author[{Zongxia Liang, Qi Ye}]
{\tiny {\tiny  Zongxia Liang$^{a}$, Qi Ye$^{b}$ }
 \vskip 10pt\noindent
{\tiny Department of
Mathematical Sciences, Tsinghua University, Beijing 100084, China }
  \footnote{\\
 $ a$ email: liangzongxia@tsinghua.edu.cn\\
 $ b$ Corresponding author, email:   yeq19@mails.tsinghua.edu.cn\\  
 Funding: This work was funded by National Natural Science Foundation of China (Grant Nos.12271290)}}
\maketitle
\noindent





\begin{abstract}
This paper delves into financial markets that incorporate a novel form of heterogeneity among investors, specifically in terms of their beliefs regarding the reliability of signals in the business cycle economy model, which may be biased. Unlike most papers in this field, we not only analyze the equilibrium but also examine welfare using objective measures while investors aim to maximize their utility based on subjective measures. Furthermore, we introduce passive investors and use their utility as a benchmark, thereby revealing the phenomenon of double loss sometimes. In the analysis, we examine two effects: the distortion effect on total welfare and the advantage effect of information and highlight their key factors of influence, with a particular emphasis on the proportion of investors. We also demonstrate that manipulating investors' estimation towards the economy can be a way to improve utility and identify an inner connection between welfare and survival.  
\end{abstract}

\section{introduction}

Is trading always beneficial for all investors in the market? The answer is no, especially in cases of homogeneous preferences where trading can lead to welfare distortion. Despite this, many investors continue to trade actively in the market under the belief that they can improve their utility. However, investors typically rely on subjective measures to evaluate their utility rather than objective measures, which they may not know. In this paper, we aim to analyze the welfare of investors in the market.

Our analysis focuses on the origin of heterogeneity in investors' beliefs about the reliability of signals in a business cycle economy model where the economy growth rate follows the mean-reverting process. We consider two classes of investors: Class-I investors, who rely on information and believe that signals are useful for more accurate estimation of the economy state, and Class-R investors, who are cautious and rely only on the actual evolution of the economy, not trusting any ``rumors" in the market. Both classes of investors observe the signals and they agree to disagree other's attitude towards the signals. However, it is important to note that these signals may be biased, and investors are not aware of the extent of the bias. Therefore, investors can only choose to trust or distrust the signals without knowing the degree of bias and even whether it is biased.

Building on the work of \cite{fedyk2013market} and \cite{he2017index}, we place significant emphasis on welfare analysis rather than just analyzing the effect of heterogeneity in the market equilibrium. While investors benefit themselves by trading in the market based on their subjective probability measure, we adopt the approach of measuring the expected ex-post welfare using objective expected utilities. Measuring welfare objectively provides a better perspective to evaluate real utility improvements. Our findings suggest that when the signal is unbiased or slightly biased, Class-I investors achieve higher utility than Class-R investors through information advantage. However, when the signal is heavily biased, the comparison is reversed as Class-R investors can improve their welfare through making use of Class-I investors's bad estimation of the economy state.

We introduce passive investors, also known as index investors, who only consume their allocated holdings of risky assets and do not participate in trading, thus having no impact on the market equilibrium. The inclusion of passive investors provides a good benchmark that the investors can achieve just through no trading, which helps us assess whether active investors can truly improve their utility by trading in the market. We found that in the unbiased case, if the proportion of Class-I investors is too large, their utility is smaller than that of passive investors even Class-I investors embrace the absolute information advantage. This implies that Class-I investors fail to achieve welfare improvements through trading, which results in a phenomenon of double loss, where all investors' welfare is lower than benchmark welfare. To illustrate this counterintuitive fact, we introduce the concept of total welfare weighted by investor wealth. We find that trading in the market is a negative-sum game in homogeneous preference in objective measure, where the double loss phenomenon can be explained by the fact that the utility improvements gained by trading with Class-R investors using information advantage are dominated by the distortion of total welfare. This phenomenon also occurs in the extremely biased case when Class-R investors are the majority. Our study highlights the proportion of each class of investors as the determining factor that measures these two effects.

We explore a scenario where hypothetical investors know the objective measure and can choose to be Class-I, Class-R, or passive investors based on their optimal strategy. Using the insights gained from this scenario, we can deduce a phenomenon. If the majority of investors are Class-R, then investors knowing the unbiased signal can take advantage of information and trade with them to improve utility. However, if the majority of investors are Class-I, taking advantage of information can lead to little or even negative utility gains. At this time, investors can manipulate the market by transforming a biased version of the signal and disguising themselves as Class-R investors to deceive Class-I investors and achieve higher utility.

At last, we delve into the survival analyze (\cite{kogan2006price}) whether investors will eventually consume little or total over time. Our findings align with the utility comparison, indicating that Class-I investors survive if the signal is unbiased or slightly biased, whereas Class-R investors survive if the signal is heavily biased. By applying techniques from repeated game theory, we demonstrate that survival and higher utility are equivalent when $\rho=0$. In other words, ignoring time preference, achieving higher utility is equivalent to survival over infinite time. However, when $\rho>0$, the differences between the two classes are affected by their early-time consumption. The other way they are equal if we allow the initial state follows the steady distribution. This assumption can be achieved if there are infinite history to learn for estimating the economy state.

In summary, this paper makes three primary contributions. First, it compares the welfare of Class-I, Class-R, and the benchmark in an objective measure, highlighting the double loss phenomenon. Second, it demonstrates that all-knowing investors can achieve utility improvements not only through their informational advantage but also by manipulating other investors' perceptions of the economic state. Last, it establishes a link between survival and higher utility. 

The paper is structured as follows: Section one presents the economic model and investors' preferences. Section two determines the market equilibrium. Section three examines the welfare of all investors. Section four investigates a hypothetical scenario, illustrating how manipulating investors' perceptions can also lead to higher utility. Section five connects survival analysis with utility analysis. Finally, the paper concludes in the last section.

\textbf{Literatures Review}: Numerous studies have researched the impact of heterogeneity on market equilibrium of a pure exchange economy, with the two most common sources of heterogeneity being differences in preferences and beliefs. In terms of preferences, investors may exhibit varying levels of time preference rates and risk aversion. Additionally, \cite{muraviev2013market} incorporates the effect of habit formation as a benchmark to provide a non-standard utility function. With regards to beliefs, investors may hold different estimates of the growth rate of the economy. The most viewed case involves just two investors with different values for the growth rate and there are amounts of paper researching on it (\cite{cvitanic2012financial} and \cite{abbot2017heterogeneous}). The most gorgeous one belong to \cite{bhamra2014asset} who even give the explicit form of the solution by considering both two sources.  

Certainly,  there are also certain mechanisms that serve as sources of heterogeneity. \cite{basak1998equilibrium}  propose the restriction of stock market participant and \cite{basak2000equilibrium} suggest the restriction of portfolio leverage for example the preclusion of short selling. \cite{ehling2018disagreement} introduce the impact of inflation, which led to disagreements among investors. Some literature even suggest that investors may have varying interpretations of the economic model, such as \cite{buraschi2006model} and \cite{cujean2017does}. For example the later consider that some investors see the state of the economy as continuously evolving throughout good and bad times, while other agents view the economy in discrete terms with good and bad times alternating.

The heterogeneity in investors' beliefs about the growth rate of the economy is far more complex than just having different certain values, as it can be influenced by various factors. \cite{david2008heterogeneous} establishes a fundamental analysis based on the model of earnings forecast dispersion. \cite{andrei2019asset} consider investors using different signals to estimate the length of the business cycle, while \cite{wang2020disagreement} consider that some investors may have a pro-cyclical belief. Academics have devoted significant effort to studying the different economic models and finding the origins of heterogeneity in the market.

In this paper, we need to emphasize the distinction between our work and \cite{he2017index}, as certain conclusions may seem similar to those presented in their paper. However, our main conclusion is derived from a completely different perspective. While their conclusions are based on the assumption of equal endowments for each class, our focus is on the proportion of each class, which we consider to be the most significant factor. Further comparative details are provided in the content.

\section{The Model}

We propose a dynamic general equilibrium model that incorporates heterogeneous investors with varying levels of confidence in estimating the business cycle through external signals. Our model includes a description of the economy and optimization problems for investors under different precision levels. We characterize the equilibrium interest rate and market price of risk within this framework.

\subsection{The Economy and Models of the Business Cycle}

We consider an economy with an aggregate dividend that flows continuously over time. The market consists of two securities: a risky asset in positive supply of one unit, and a riskless asset in zero net supply. The risky asset, which we call the stock, is a claim to the dividend process, $D$, that satisfies the stochastic differential equation:
\begin{equation*}
dD_t = \mu_t D_t dt + \sigma_D D_t dW_t,
\end{equation*}
where $W$ is a standard Brownian motion under the objective probability measure $P$. The expected dividend growth rate, denoted $\mu$, is unobservable and follows a mean-reverting process:
\begin{equation*}
	d\mu_t = \kappa (\bar{\mu} -\mu_t) +\sigma_\mu dW^\mu_t,
\end{equation*}
where $\bar{\mu}$ is the long-term mean, $\kappa$ is the speed of reverting and Brownian motion $W^\mu$ is independent to $W$ in objective measure.

The economy is populated by two classes of investors, R and I, who consume dividends and trade in the market. As the state of the economy is unobservable, investors need to estimate it using the empirical realization of dividends. Class-R investors estimate the state only via the process of dividends D, while Class-I investors hear information in the market and believe in an external signal process:
\begin{equation*}
de_t = \mu_t dt + \sigma_e dB^\zeta_t,
\end{equation*}
where $B^\zeta$ is also a standard Brownian motion in their own subjective probability measure $P^I$ and independent of W and $W^\mu$ . However, $B^\zeta$ is not a Brownian motion in the objective probability measure P. Class-I investors incorrectly believe this signal is unbiased for estimating $\mu_t$, while the signal may be processed incorrectly, for instance, due to fallacious data collection, inaccurate analyst reports, or it may even be deliberately leaked by Class-R investors, where more details will be represented later. In fact, the signal follows:
\begin{equation*}
de_t = (\mu_t + \sigma_e \zeta) dt + \sigma_e dB_t,
\end{equation*}
where $\zeta$ is a constant and B is a Brownian motion in the objective probability space. If $\zeta=0$, the signal is unbiased. We can see the relation of $B^\zeta$ and B with:
\begin{equation*}
dB^\zeta_t = \zeta dt + dB_t.
\end{equation*}

It must be stated that Class-R investors also observe the external signal. However, for prudence, they do not trust the signal and will not estimate the market state according to it. Each class of investors agree to disagree the other's attitudes towards the signal. This will affect the equilibrium as it will make Class-R investors aware of the estimation held by Class I. 

Furthermore, the precision of the external signal is reflected by the inverse of the diffusion parameter, denoted as:
\begin{equation*}
h_e := 1/\sigma_e.
\end{equation*}
As $h_e$ approaches zero, the signal becomes less informative and contributes little to the estimation. Conversely, when $h_e$ approaches infinity, the investors have complete knowledge of the drift rate at any time.

Likewise, the precision of the ``dividend signal" is given by:
\begin{equation*}
h_D := 1/\sigma_D.
\end{equation*}

Investors update their expectations based on their sources of information and derive an estimate, which we refer to as the filter. Assuming that the initial estimation of Class-R investors follows normal distribution with mean $\mu^R_0$ and variance $\gamma^R_0$ and the initial estimation of Class-I investors follows normal distribution with mean $\mu^I_0$ and variance $\gamma^I_0$.  Proposition below presents the dynamics of the filters for classes R and I of investors.

\begin{proposition}
 The filter of Class-R investors, the conditional expectation $\mu^R_t = E_t^R[\mu_t]$ and the mean square error $\gamma^R_t = E[(\mu_t-\mu^R_t)^2]$ , evolves according to 
 \begin{equation*}\label{filterR}
 	d\mu^R_t = \kappa(\bar{\mu}-\mu^R_t)dt + \gamma^R_t h_D d\hat{W}^R_t,
 \end{equation*}
 \begin{equation*}
 	d\gamma^R_t = \left( -2\kappa \gamma^R_t +\sigma_\mu^2 - h_D^2 (\gamma^R_t)^2    \right)dt,
 \end{equation*}
 where $\hat{W}^R$ is a Brownian motion under class R's probability measure $P^R$,
 \begin{equation*}
 	d\hat{W}^R_t = h_D (dD_t/D_t - \mu^R_t dt).
 \end{equation*}
 Similarly, the filter of Class-I investors, the conditional expectation $\mu^I_t = E_t^I[\mu_t]$ and the mean square error $\gamma^I_t = E[(\mu_t-\mu^I_t)^2]$, evolves according to 
 \begin{equation*}\label{filterI}
 	d\mu^I_t = \kappa(\bar{\mu}-\mu^I_t)dt + \gamma^I_t (h_D d\hat{W}^I_t +h_e d\hat{B}^I_t),
 \end{equation*}
 \begin{equation*}
 	d\gamma^I_t = \left( -2\kappa \gamma^I_t +\sigma_\mu^2 - (h_D^2+h_e^2) (\gamma^I_t)^2    \right)dt,
 \end{equation*}
 where $\hat{W}^I$ and $\hat{B}^I$ are Brownian motion under class I's probability measure $P^I$,
 \begin{equation*}
 	d\hat{W}^I_t = h_D (dD_t/D_t - \mu^I_t dt),
 \end{equation*}
 \begin{equation*}
 	d\hat{B}^I_t = h_e (de_t - \mu^I_t dt).
 \end{equation*}

\end{proposition}

\begin{proof}
	See \cite{liptser2013statistics}. 
\end{proof}

As we can see, when $h_e \to 0$, the case of I reduces to the case of R.

If the real state is known (which will be used in welfare analysis), then we can rewrite the processes $\mu^R$ and $\mu^I$ as follows:
\begin{equation*}
d\mu^R_t = \kappa(\bar{\mu}-\mu^R_t)dt +\gamma^R_t h_D^2 (\mu_t-\mu^R_t)dt + \gamma^R_t h_D dW_t,
\end{equation*}
\begin{align*}\nonumber
d\mu^I_t &=  \kappa(\bar{\mu}-\mu^I_t)dt +\gamma^I_t (h_D^2+h_e^2) (\mu_t-\mu^I_t)dt\\
& \ \ \ +\gamma^I_t h_e \zeta dt+ \gamma^R_t (h_D dW_t+h_e dB_t).
\end{align*}

For convenience, when no specific analysis of $\zeta$ is considered, we will write $\mu^I_t$ as $\mu^I_t(\zeta)$ throughout the remainder of this discussion.

From the Proposition, we can see that the drift term $\mu^R$ is determined by two factors: the mean-reverting term $\kappa(\bar{\mu}-\mu^R_t)$ and the term toward the real economy $\gamma^R_t h_D^2 (\mu_t-\mu^R_t)$. The drift term $\mu^I$ is determined by three factors: the first two terms are the same as $\mu^R$, and the biased term $\gamma^I_t h_e \zeta$.

Additionally, when $\zeta=0$, due to the fact that $(h_D^2+h_e^2)>h_D^2$ in the second term, indicating that external signals facilitate faster learning towards the real economic state. However, if $\zeta>0$, the signal is biased towards optimistic estimates, leading to an overestimation of class I, whereas if $\zeta<0$, the opposite is true.

The following figures illustrate the evolution of the real economic state and their estimates. One plot depicts the unbiased case, while the other two represent the optimistic and pessimistic cases. We observe that in the unbiased case, $\mu^I_t$ is usually closer to $\mu_t$. However, in the biased cases, we observe some deviations in the estimation of $\mu^I$.

The parameter values used are given that $\bar{\mu}=0.04$, $\kappa=0.2$, $\sigma_\mu=0.01$, $\sigma_D=0.2$, $\sigma_e=0.05$, $\rho=0.02$, and initial condition as $\mu^R_t=\mu^I_t=\mu_0=\bar{\mu}$ and $\gamma^R_0 = \gamma^I_0 =0$ . The latter figure employs the same parameters as the former.

\begin{figure}
  \centering
  \begin{subfigure}[ht]{0.9\textwidth}
    \includegraphics[width=\textwidth]{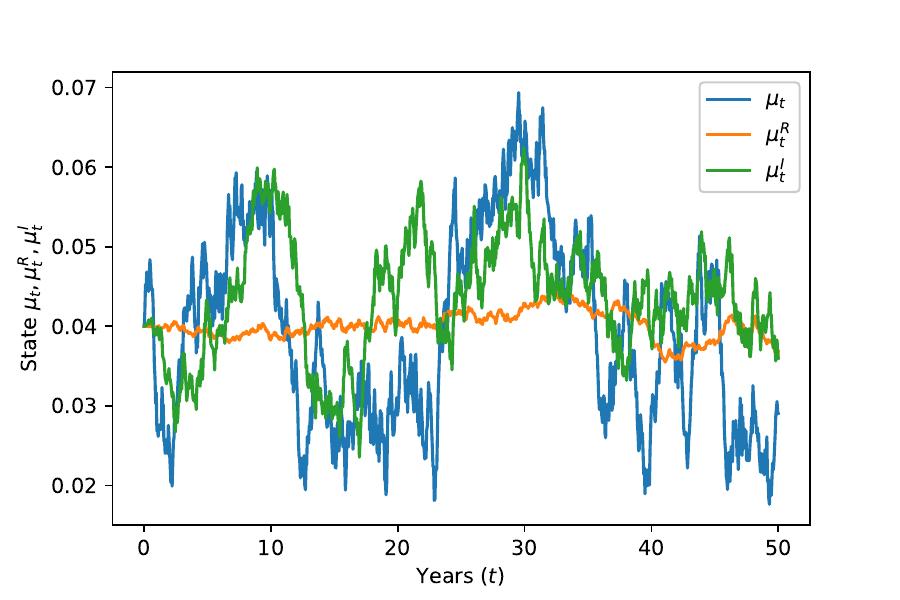}
    \caption{$\zeta=0$ .}
  \end{subfigure}
  \centering
  \begin{subfigure}[h]{0.48\textwidth}
    \includegraphics[width=\textwidth]{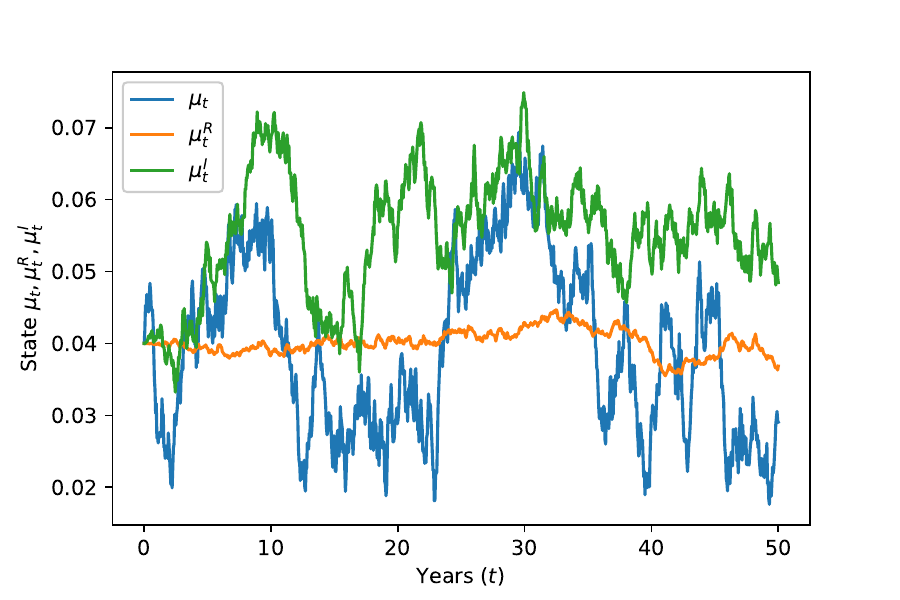}
    \caption{$\zeta=1$ .}
  \end{subfigure}
  \hfill
  \begin{subfigure}[h]{0.48\textwidth}
    \includegraphics[width=\textwidth]{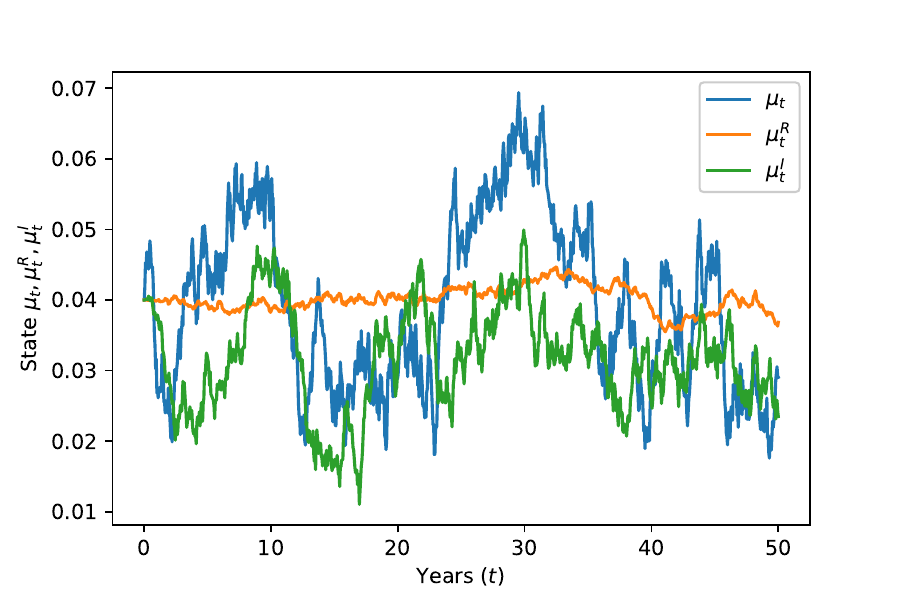}
    \caption{$\zeta=-1$ .}
    \label{fig:figure2}
  \end{subfigure}
  \caption{The evolution of $\mu^I$ and $\mu^R$.}
\end{figure}

\subsection{Investors' Preference}

The objective of investors in each class is to maximize their lifetime utility through trading in the market, given by

\begin{equation*}
E^m\bigg[ \int_0^\infty u(s,c^m_s) ds \bigg], \ \ m = R, I,
\end{equation*}
where $u(s,x):=e^{-\rho s} \log(x)$, $\rho$ is the time discount factor, and $E^m$ denotes the expectation under the subjective beliefs $P^m$ of the Class-$m$ investor. The wealth process is given by

\begin{equation}\label{wealth}
dX^m_t =\left[ r_t X^m_t +\pi^m_t ({\mu_P}_t +\frac{D_t}{P_t}-r_t)X^m_t -c^m_t\right]dt+\pi^m_t X^m_t{\sigma_P}_t dW_t,
\end{equation}
where $r$ is the interest rate and $P$ is the risky asset's price, determined through market equilibrium.

\section{Market Equilibrium}
A rational expectations equilibrium (REE) refers to a state in which each agent in the economy maximizes their lifetime utility, and a set of conjectured prices for all securities is established for each date and state for each agent, such that total consumption equals total output in the economy, markets clear, and agents agree on prices in all dates and states.

Using the martingale method, we can reduce equation (\ref{wealth}) to an inequality as follows:

\begin{equation}\label{martingale}
E^m\left[ \int_0^\infty c^m_s \xi^m_s ds\right]\leq E^m\left[ \int_0^\infty e^m x_s \xi^m_s ds \right] \equiv X^m_0,
\end{equation}
where $e^R$ ($e^I$) denotes the Class-R's (Class-I's) endowment with a proportion of claim of the economy output, such that $e^R+e^I=1$, $X^m_0$ represents the value of their endowment, and $\xi^m_t$ denotes the Class-$m$ investor's state-price density (SPD) function for consumption at time $t$.

We write (\ref{martingale}) in terms of the SPD function as:

\begin{equation}\label{SPD}
\frac{d\xi^m_t}{\xi^m_t}=-r_t dt- \varphi^m_t d\hat{W}^m_t,
\end{equation}
where $r_t$ represents the real rate of interest, and $\varphi^m_t$ denotes the market price of risk for the Class-$m$ investor.

The first order condition can then be expressed as:

\begin{equation*}
u_x(t,c^m_t)=y^m \xi^m_t,
\end{equation*}
where $y^m$ is determined by the constraint of (\ref{martingale}).

As the investors must agree on the level of prices at all dates (see in \cite{david2008heterogeneous}), this condition must be met:

\begin{equation*}
\varphi^R_t - \varphi^I_t = h_D(\mu^R_t-\mu^I_t ).
\end{equation*}

A key state variable in the analysis is the disagreement value process:

\begin{equation*}
\eta_t=\frac{\xi^I_t}{\xi^R_t},
\end{equation*}
which is the ratio of Class-I investors' SPD to the Class-R investors' SPD. $\eta_t$ is the ratio of the likelihood of observing the fundamentals at date $t$ as a realization of the model of Class-I investors to that of the Class-R investors. Then we calculate that:

\begin{equation*}
\frac{d\eta_t}{\eta_t}=(\mu^R_t-\mu^I_t )h_D d\hat{W}^I_t.
\end{equation*}

The process $\eta$ enables us to study the evolution of the disagreement process given the history of each class's beliefs.

\begin{theorem}\label{theorem1}
	1.The individual consumption flow rates are 
	\begin{equation*}
  \left\{
  \begin{aligned}
    & c^R_t=\lambda_t D_t \\
    & c^I_t=(1-\lambda_t)D_t,
  \end{aligned}
  \right.
\end{equation*}
	where
	\begin{equation*}
		\lambda_t= \frac{k\eta_t}{1+k\eta_t},
	\end{equation*}
	and $k=\frac{y^I}{y^R} $.
	
	2. The interest rate and market prices of risk in the economy is given with 
	
	\begin{equation*}
  \left\{
  \begin{aligned}
    & r_t=\rho +\left(\lambda_t \mu^R_t+(1-\lambda_t)\mu^I_t \right) -\sigma_D^2 \\
    & \varphi^R_t= \sigma_D +h_D (1-\lambda_t)( \mu^R_t-\mu^I_t ) \\
    & \varphi^I_t= \sigma_D - h_D \lambda_t( \mu^R_t-\mu^I_t ).
  \end{aligned}
  \right.
\end{equation*}
\end{theorem}

\begin{proof}
	Using the equation $c^m_t=(e^{\rho t}y^m \xi^m_t)^{-1 }$ and the process of $\xi_t^m$ (\ref{SPD}) we obtain 
	\begin{align*} \nonumber
		dc^m_t &=(r_t-\rho + (\varphi^m_t)^2)c^m_t dt + \varphi^m_t c^m_t dW^m_t\\
		&=(r_t-\rho + (\varphi^m_t)^2+\varphi^m_t (\mu^t-\mu^m_t)h_D )c^m_t dt + \varphi^m_t c^m_t dW_t,
	\end{align*}
	for $m=R, I$.
	As $c^R_t+c^I_t=D_t$, then there will be the relation: $dD_t = dc^R_t+dc^I_t$. Comparing two process's drift term and diffusion term. We can obtain the interest rate and the market price of risk of each class.
\end{proof}

As for the interest rate, the first term represents time preference, the second term reflects the usual wealth effect on consumption. Specifically, when the growth rate of consumption increases, investors become less willing to save for the future, leading to a higher equilibrium interest rate. Here, the growth rate is the average expected growth rate of investors, weighted by their shares of total consumption. The third term reflects the precautionary demand that arises from the volatility in the process. If the volatility is larger, investors will have more demand to offset this risk, which will lower the equilibrium interest rate.

\section{Welfare analysis}

Previous literature has largely focused on the exchange economy setting with various forms of heterogeneity, resulting in speculation, redistribution, and inefficient risk sharing. In contrast, our research focuses on welfare analysis. While participants may benefit from trading in the market, these benefits are only under their subjective probability measure. Thus, it is important to compare each class's utility under the objective probability measure to determine whether all investors benefit from trading.

For convenience, we define the welfare $U^m$ of Class-$m$ as given by
\begin{equation*}
U^m:=E\bigg[ \int_0^\infty u(t,\frac{c_t}{X_0^m}) dt \bigg].
\end{equation*}
It represents the utility achieved under the objective probability measure with one unit endowment of Class-$m$.

The introduction of the objective probability measure has led to some adjustments to \textbf{Theorem \ref{theorem1}}.

There is a unique real state price density $\xi_t$ with an initial value of $\xi_0=1$, and
\begin{equation*}
\frac{d\xi_t}{\xi_t}=-r_t dt- \varphi_t dW_t,
\end{equation*}
where $\varphi_t$ denotes the real market prices of risk at time $t$.

Similarly, we obtain:
\begin{equation*}
\varphi_t-\varphi^m_t= h_D(\mu_t-\mu^m_t).
\end{equation*}

This is because the price in the market must be the same. From \textbf{Theorem \ref{theorem1}}, we have:
\begin{equation*}
\varphi_t= \sigma_D +h_D ((1-\lambda_t)(\mu_t-\mu^I_t)+\lambda_t (\mu_t-\mu^R_t)).
\end{equation*}

Next, we characterize Class-$m$'s subjective probability measure as follows:

\begin{equation*}
\eta^m_t=\frac{\xi_t}{\xi^m_t},
\end{equation*}
which evolves as:
\begin{equation*}
\frac{d\eta^m_t}{\eta^m_t}=(\mu^m_t-\mu_t)h_DdW_t.
\end{equation*}

We will analyze the equilibrium and welfare based on the perspective of the objective probability measure.

\subsection{Comparison of Class-I's welfare and Class-R's welfare}
We begin by considering the case where the external signal is unbiased, i.e., $\zeta=0$. In this case, the Class-I investors truly benefit from their information advantage. We have the following theorem:

\begin{theorem}
In the unbiased case, assuming $\mu^R_0=\mu^I_0=\mu_0$ and $\gamma^R_0 = \gamma^I_0$ to avoid the effect of the initial estimation of each class, we have
\begin{equation*}
U^I\geq U^R,
\end{equation*}
where $U^I$ and $U^R$ are the utilities of Class-I and Class-R investors, respectively.
\end{theorem}

\begin{proof}	
The martingale relation, which holds for the optimal consumption state, is given in the objective probability measure by:

\begin{equation*}
E\left[\int_0^\infty c^m_t \xi_t dt \right]= X^m_0
\end{equation*}
and
\begin{equation*}
u_x(t,c^m_t)=y^m \xi^m_t=y^m \xi_t / \eta^m_t.
\end{equation*}
The above equation can be simplified as:
\begin{equation*}
c^m_t = e^{-\rho t} \frac{\eta^m_t}{y^m \xi_t },
\end{equation*}

\begin{equation*}
X^m_0=E\left[\int_0^\infty e^{-\rho t} \frac{\eta^m_t}{y^m \xi_t } \xi_t dt \right]=\frac{1}{y^m} \int_0^\infty e^{-\rho t}E\left[ \eta^m_t \right] dt=\frac{1}{\rho y^m},
\end{equation*}
and
\begin{equation*}
U^m=E\left[ \int_0^\infty u(t,\frac{c_t}{X_0^m}) dt \right]=E\left[ \int_0^\infty e^{-\rho t} \log\left(e^{-\rho t}\frac{\rho \eta^m_t}{\xi_t } \right) dt \right].
\end{equation*}
Then we have:
\begin{equation*}
U^I-U^R = E\bigg[ \int_0^\infty e^{-\rho t} (\log(\eta_t^I)-\log(\eta_t^R)) dt \bigg].
\end{equation*}
Moreover, we have:
\begin{equation*}
\eta^m_t = \exp\bigg( -\frac{1}{2}\int_0^t (\mu^m_s-\mu_s)^2h_D^2 ds +\int_0^t(\mu^m_s-\mu_s)h_D dW_s \bigg).
\end{equation*}
Then, we write:
\begin{equation*}
E\bigg[ \log(\eta_t^I)-\log(\eta_t^R) \bigg] = \frac{h_D^2}{2}E\bigg[ \int_0^t (\mu^R_s-\mu_s)^2- (\mu^I_s-\mu_s)^2 ds \bigg].
\end{equation*}
If we rewrite $\mu^R_s=E[\mu_s|\mathcal{F}^D_s]$ and $\mu^I_s=E[\mu_s|\mathcal{F}^D_s \vee\mathcal{F}^e_s]$, then in the perspective of conditional expectation, both $\mu^R_s$ and $\mu^I_s$ are $\mathcal{F}^D_s \vee\mathcal{F}^e_s$-measurable. Thus, we have:
\begin{equation*}
E[(\mu^R_s-\mu_s)^2]\geq E[(\mu^I_s-\mu_s)^2].
\end{equation*}
Using this relation, we have:
\begin{equation*}
U^I \geq U^R.
\end{equation*}
\end{proof}

The comparison of the utility of each class of investors depends on their ability to accurately estimate the state of the economy. If the real economy is in the high state, Class-I investors are likely to be more optimistic than Class-R investors, resulting in Class-I borrowing money from Class-R to purchase more risky assets. In this scenario, Class-I investors stand to gain more from the leverage as the risky assets are expected to generate higher dividends. The reverse case holds true when the economy is in the low state.

In the biased case, $\zeta$ is not equal to zero, and we need to consider the effect of a biased signal on the estimation. It is intuitive to assume that a biased signal may negatively impact the utility of Class-I investors. They cannot benefit from an absolute information advantage over Class-R investors as the biased signal may cause Class-I investors to make suboptimal strategies due to inaccurate estimation.

As previously mentioned, $\mu_s^I(\zeta)$ is written as $\mu_s^I$ if $\zeta$ is not considered in the analysis. This applies to other representations such as $U^m$ for $U^m(\zeta)$. Then we have the following theorem.

\begin{theorem}\label{theorem2}
Assuming $\mu^R_0=\mu^I_0=\mu_0$ and $\gamma^R_0 = \gamma^I_0$, then there exist constants $\zeta_1<0$ and $\zeta_2 >0$ such that if $\zeta\in (\zeta_1,\zeta_2)$, then $U^I(\zeta) > U^R(\zeta)$, and if $\zeta\in (-\infty, \zeta_1)\cup(\zeta_2, \infty )$, then $U^I(\zeta) < U^R(\zeta)$.
\end{theorem}

\begin{proof}
	See Appendix A.
\end{proof}

\begin{remark}
	To eliminate the influence of initial estimation, we have assumed that $\mu^R_0=\mu^I_0=\mu_0$ and  $\gamma^R_0 = \gamma^I_0$. In fact, in the case where there is no bias, the assumption that $\mu^R_0=\mu^I_0=\mu_0$ can be relaxed into $\mu^R_0=\mu^I_0$ and all conclusions would still hold. However, in the presence of bias, this condition cannot be relaxed. As an example, consider a case with parameter given $\zeta=10$, $\mu^R_0=\mu^I_0= \bar{\mu}$, $\mu_0=\bar{\mu}+10$ and $\rho=100$. It is important to note that $U^R$ may not necessarily be greater than $U^I$ because the optimistic signal can cause Class-I investors to estimate the economy state more optimistically at the beginning of the time, thereby resulting in more accurate estimations of the real economy.
\end{remark}

It is important to mention another detail beforehand: if we consider them as functions of $e^R$ (which will be used in a later subsection), the difference between them remains constant.

\subsection{Comparison of each class's welfare and benchmark welfare}

Additionally, we are not only comparing the utility between Class-R and Class-I, but also introducing the benchmark utility, which is achieved with no trading in the market. This benchmark is set by passive investors who only consume the dividend from the risky asset. By comparing the welfare of both Classes with the benchmark utility, we can see whether both Classes can actually benefit themselves through trading. The welfare $U^{bench}$ is defined as the expected value of the integral over time of the utility function $u$ with respect to the fraction of the risky asset's dividend, as follows:

\begin{equation*}
U^{bench}:=E\bigg[ \int_0^\infty u(t,\frac{D_t}{P_0}) dt \bigg].
\end{equation*}
Firstly, we present the following lemma:
\begin{lemma}\label{lemma1}
	$U^{bench}$ is constant for any $\zeta$.
\end{lemma}
\begin{proof}
See Appendix B.
\end{proof}

The exact relation between $U^I$, $U^R$, and $U^{bench}$ is not fixed. It is possible to conjecture that the benchmark welfare $U^{bench}$ lies between $U^I$ and $U^R$ as passive investors do not gain or lose from information. However, it is not always the case that both $U^I<U^{bench}$ and $U^R<U^{bench}$ exist simultaneously, even in situations where Class-I investors have an absolute information advantage. Here we introduce the most dominant factor, which is the initial proportion of each class. The following theorem describes the relationship between $U^{I}$, $U^{R}$, and $U^{bench}$ with respect to changes in the initial proportion of Class-R investors.

\begin{theorem}\label{theorem3}
	In the unbiased case, assuming that $\mu^R_0=\mu^I_0=\mu_0$ and $\gamma^R_0 = \gamma^I_0$. If the initial proportion of Class-R investors $e^R$ increases, $U^{bench}$ remains constant, but $U^R$ and $U^I$ initially decrease and then increase. Mathematically, if we treat $U^R$ and $U^I$ as functions of $e^R$, then $U^{I'}(0)< 0$, $U^{I'}(1)> 0$, and $U^{I}(e^R)$ is convex. The property still exists for $U^R$ as the gap between $U^I$ and $U^R$ is constant. 
	
	Furthermore, the conclusion still holds for the biased case.
\end{theorem}

\begin{proof}
See Appendix C.
\end{proof} 

Here are some figures representing their utility. As we can see, it is evident that when the value of $e^R$ is in close proximity to zero, the welfare of both investors falls below the benchmark welfare. Additionally, as the value of $\zeta$ increases or decreases, it becomes apparent that $U^R>U^I$, and the occurrence of double loss persists when $e^R$ is in close proximity to one.
\begin{figure}
  \centering
  \begin{subfigure}[ht]{0.8\textwidth}
    \includegraphics[width=\textwidth]{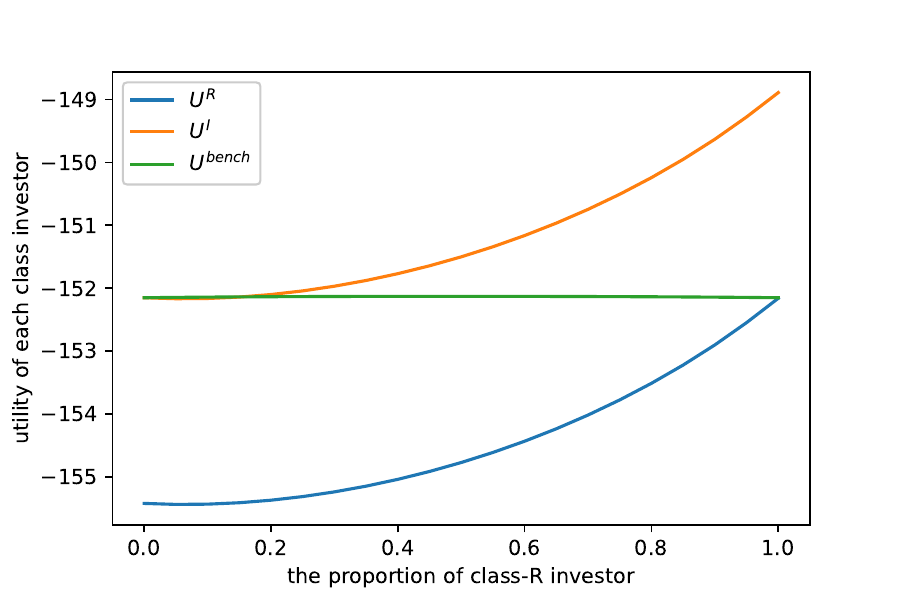}
    \caption{$\zeta=0$ .}
  \end{subfigure}
  \centering
  \begin{subfigure}[ht]{0.48\textwidth}
    \includegraphics[width=\textwidth]{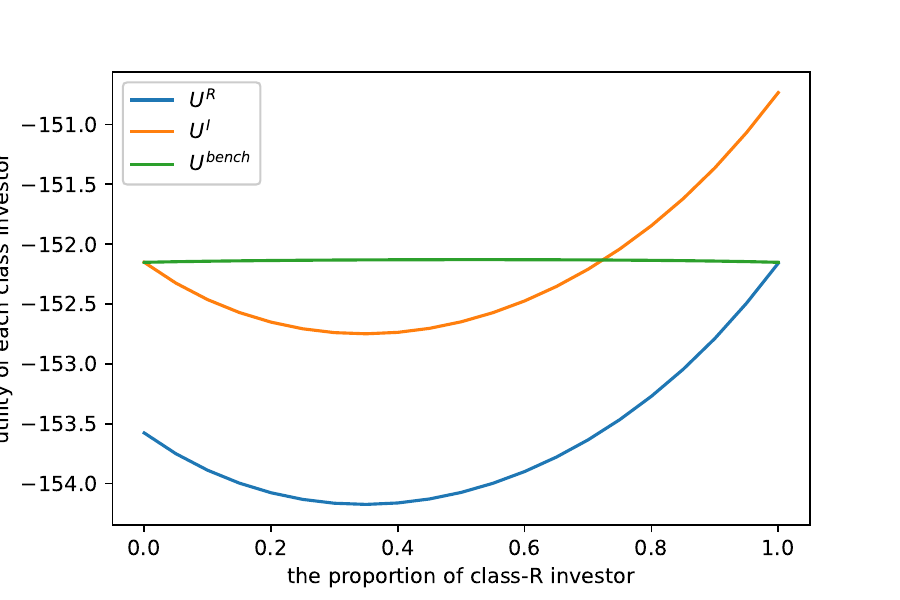}
    \caption{$\zeta=0.5$ .}
  \end{subfigure}
  \hfill
  \begin{subfigure}[ht]{0.48\textwidth}
    \includegraphics[width=\textwidth]{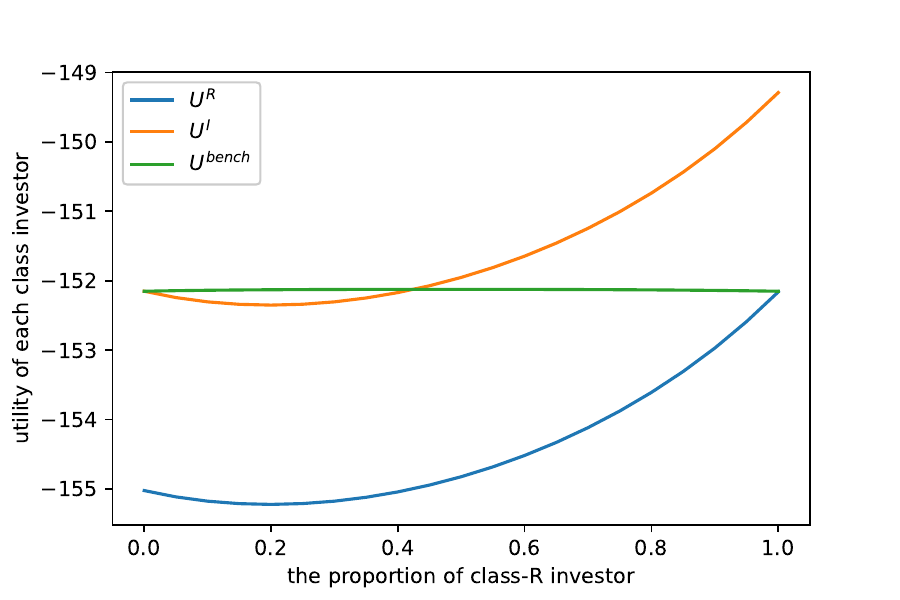}
    \caption{$\zeta=-0.5$ .}
  \end{subfigure}
  \centering
  \begin{subfigure}[h]{0.48\textwidth}
    \includegraphics[width=\textwidth]{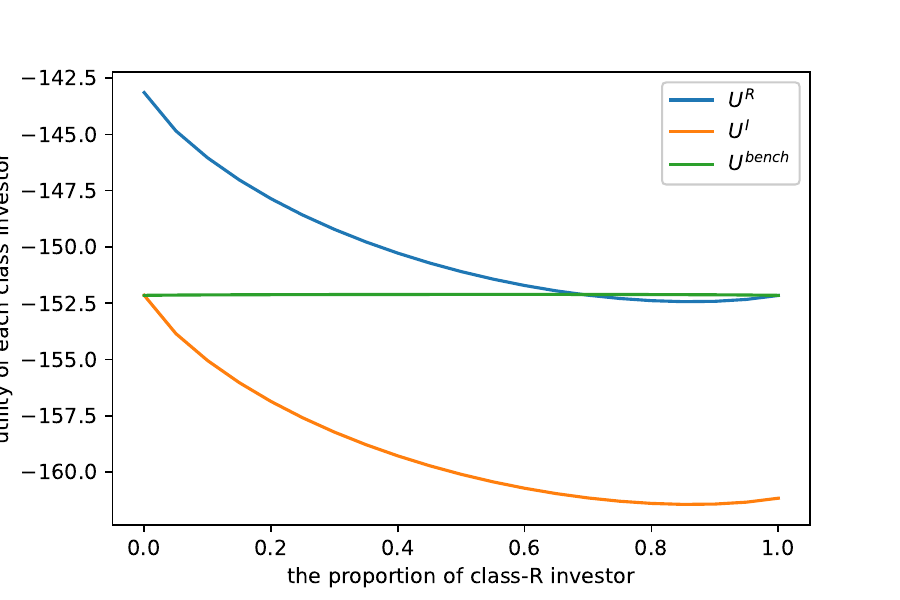}
    \caption{$\zeta=1.5$ .}
  \end{subfigure}
  \hfill
  \begin{subfigure}[h]{0.48\textwidth}
    \includegraphics[width=\textwidth]{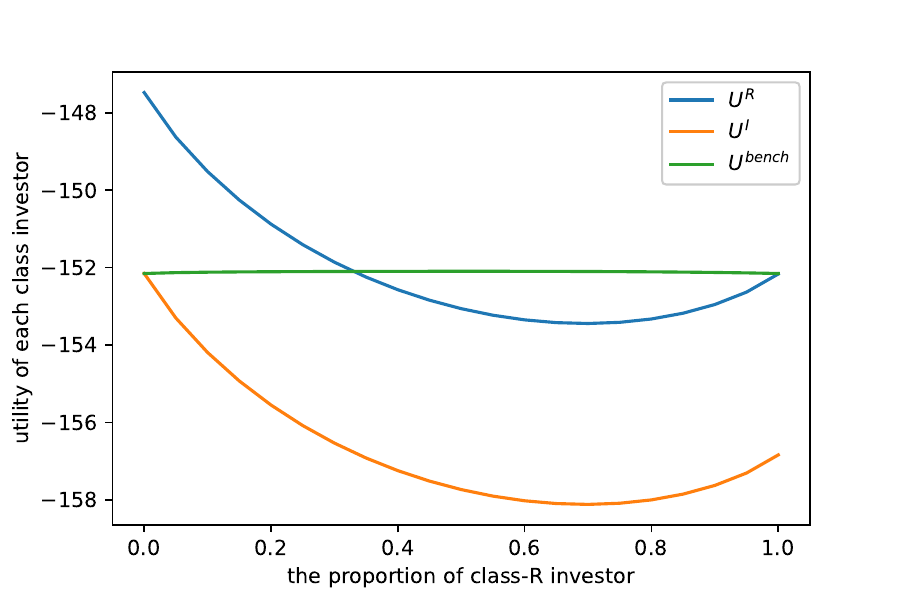}
    \caption{$\zeta=-1.5$ .}
  \end{subfigure}
  \centering
  \begin{subfigure}[h]{0.48\textwidth}
    \includegraphics[width=\textwidth]{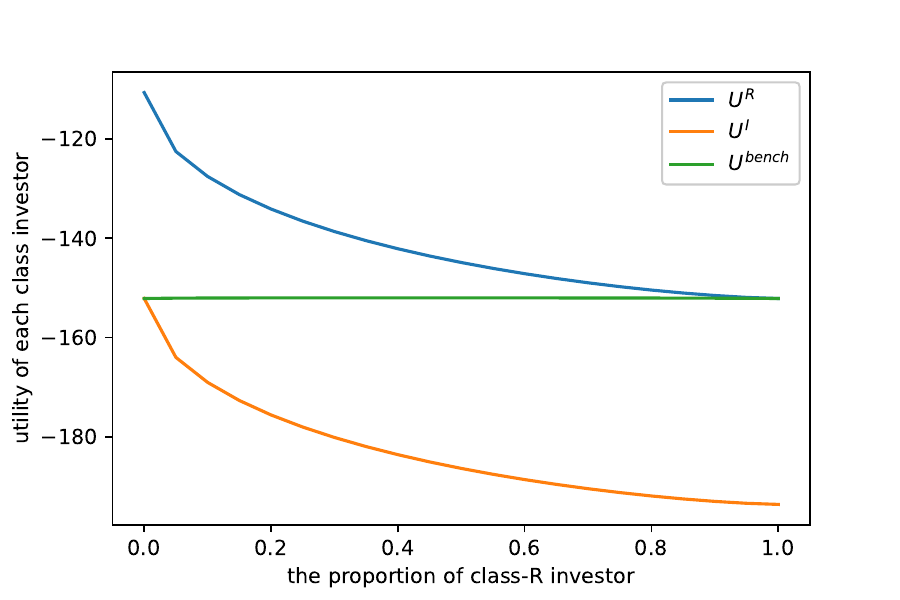}
    \caption{$\zeta=3$ .}
  \end{subfigure}
  \hfill
  \begin{subfigure}[h]{0.48\textwidth}
    \includegraphics[width=\textwidth]{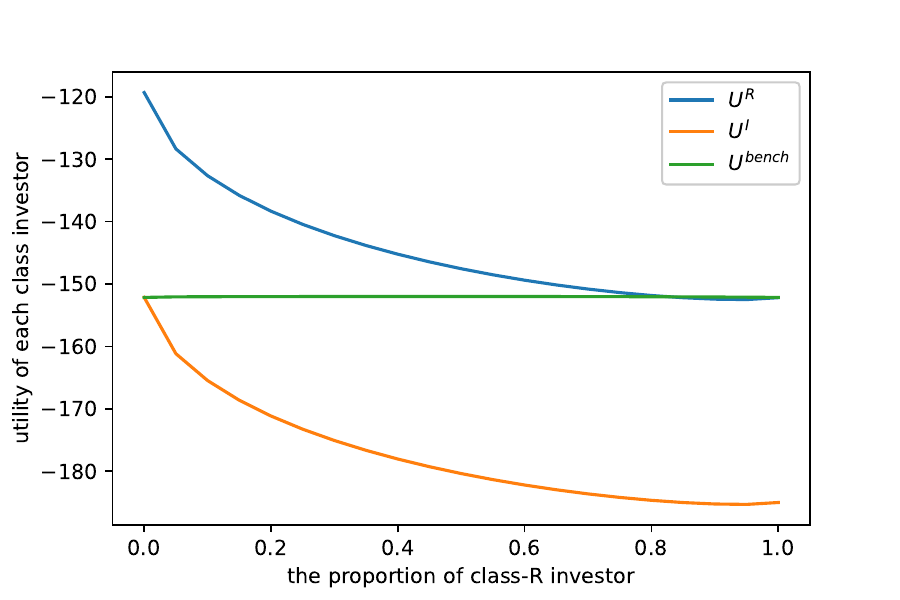}
    \caption{$\zeta=-3$ .}
  \end{subfigure}
\caption{welfare analysis in mean-reverting model.}
\end{figure}

It is counterintuitive to discover that $U^I$ does not always increase in the unbiased case with more Class-R investors in the market, with whom the Class-I investors can trade speculatively and take advantage. There is no doubt that this advantage exists, but it is careless to ignore the effect of market inefficiency.

Here, we define $U^{total}$ as:
\begin{equation*}
U^{total}:= e^R E\bigg[\int_0^\infty u(t,\frac{c_t^R}{e^R})dt\bigg] +e^I E\bigg[\int_0^\infty u(t,\frac{c_t^I}{e^I})dt\bigg]
\end{equation*}
which measures the total welfare of all investors in the market. Assuming there exists a central planner whose objective is to maximize $U^{total}$ with the constraint $c^R_t+c^I_t=D_t$, the maximization is achieved if $c^R_t=e^R D_t$ and $c^I_t=e^I D_t$. In other words, no trading in the market and investors only consuming their initial proportion of dividend is optimal for total welfare. Unfortunately, the equilibrium in the model where investors trade to improve utility in their own subjective probability measure is not optimal for total welfare, except when there is only one class of investors, i.e., $e^R=0$ or $e^R=1$.

Furthermore, we can calculate that $U^{total} = e^IU^{I}+ e^RU^{R}+ C$, where $C$ is a constant with respect to $e^R$. Thus, $U^{I'}(e^R)= (U^I-U^R)-\frac{dU^{total}}{de^R}$. Therefore, when the market is dominated by Class-I investors, and few Class-R investors enter the market, the total welfare decreases and market efficiency is damaged. This effect $\frac{dU^{total}}{de^R}$ dominates the slight welfare improvement $(U^I-U^R)$ from the advantage taken from a small number of Class-R investors, and Class-I investors' utility decreases, i.e., $U^{I'}(0)< 0$.

Similarly, in the extremely biased case where Class-R investors can trade speculatively and take advantage of Class-I's inaccurate estimation, but few Class-I investors are present in the market, the utility of Class-R investors' welfare is damaged as more Class-I investors enter the market due to market inefficiency. By the symmetry, we have $U^{R'}(e^I)= (U^R-U^I)-\frac{dU^{total}}{de^I}$ (note that we write it as a function of $e^I$) and $(U^R-U^I)|_{e^I=0} < \frac{dU^{total}}{de^I}|_{e^I=0}$. That is to say, the effect of market efficiency also dominates the welfare improvement of Class-R investors through trading.

Overall, the discussion of total welfare implies that trading in the homogenous preference case is a negative-sum game. Although trading speculatively through others' inaccurate estimation can truly gain welfare improvement, the improvement may not cover the loss from market inefficiency. That is to say, in some situations, the welfare of both classes is less than the benchmark welfare that investors can achieve through no trading in the market. This outcome is a double loss for all.

\begin{remark}
	Here we explore the differences in conclusions compared to (\cite{he2017index}), where there are some similarities in the results but the premise are totally different. In their paper, investors hold biased beliefs of constant growth rates, and this bias is correlated with the amounts of risky assets they possess. The reason why their conclusion about investors' utility not exceeding passive investors' utility is that when the bias correlation is closer to $-1$, the ``average" beliefs of passive investors align more closely with the real economy's growth rate.

In contrast, our model features a unique risky asset, and the Class-I investors' estimations of the growth rate are absolutely closer to those of Class-R investors due to the relationship of the conditional expectation. Our conclusion is based on considering the initial proportion, whereas their conclusion is derived under the assumption that all investors have the same endowment.
\end{remark}

\begin{remark}
Here, we aim to examine the generalizability of the conclusions drawn in the study. For instance, the findings still hold even if the bias degree $\zeta$ is not a constant and can be any specific process, as long as it is independent of the economy process and signal process. A determined process can satisfy this condition. Additionally, it is not necessary for the signal process to be present in the market at all times. For example, a truncated signal where $e_t \equiv 0, t\geq T$ for a certain duration of time $T$, does not affect the conclusions drawn.
	
Nevertheless, it is important to note that all conclusions stated in  \textbf{Theorem \ref{theorem3}} may not hold for any economic models. For instance, let's consider another widely used business cycle model as an example. The expected dividend growth rate $\mu$ is unobservable and follows a two-state continuous Markov chain:
\begin{equation*}
	\mu_t \in \{ \mu^h, \mu^l \} \text{ with generator matrix }
	\begin{pmatrix}
	-\lambda & \lambda \\
	\psi & -\psi
	\end{pmatrix} 
\end{equation*}
where $\lambda>0$ and $\psi>0$ are the intensities of moving from the high to low state and from the low to high state, respectively.

By following a similar deductive approach as presented in \cite{veronesi2000does} or \cite{cujean2017does}, we calculate the processes $\mu^R$ and $\mu^I$ as follows:
\begin{equation*}
d\mu^R_t = (\lambda+\psi)(\mu_\infty-\mu^R_t)dt + h_D^2\nu(\mu^R_t)(\mu_t-\mu^R_t) dt+\nu(\mu^R_t)h_D dW_t,
\end{equation*}
\begin{align*}\nonumber
d\mu^I_t &= (\lambda+\psi)(\mu_\infty-\mu^I_t)dt + (h_D^2+h_e^2)\nu(\mu^I_t)(\mu_t-\mu^I_t) dt\\
& \ \ \ +\nu(\mu^I_t)h_e \zeta dt+\nu(\mu^I_t)(h_D dW_t+h_e dB_t),
\end{align*}
where $\nu(\cdot):= (\mu^h -\cdot) (\cdot-\mu^l)$.

\begin{figure}
  \centering
    \includegraphics[width=\textwidth]{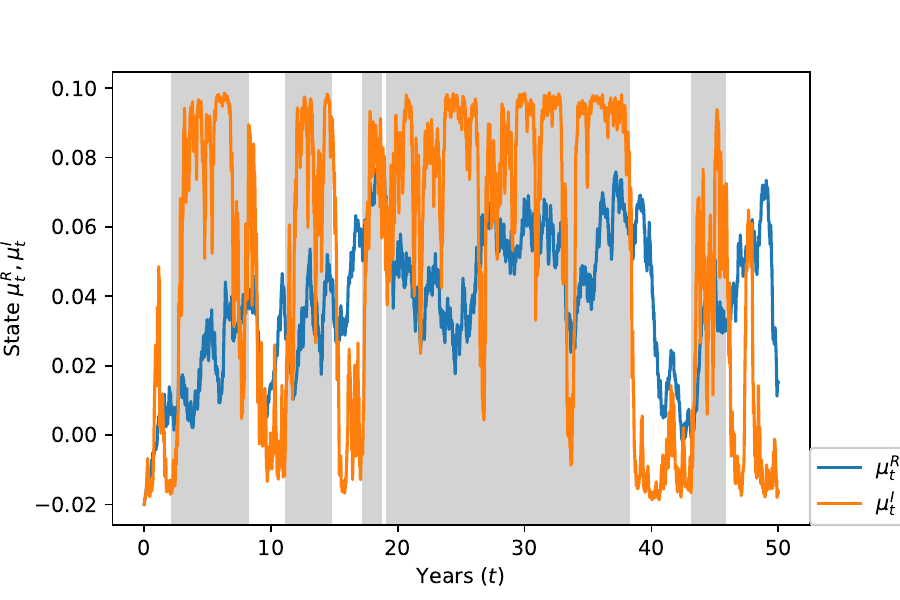}
    \caption{estimation evolution in state altering model.}
 \end{figure}

\begin{figure}
  \centering
    \includegraphics[width=\textwidth]{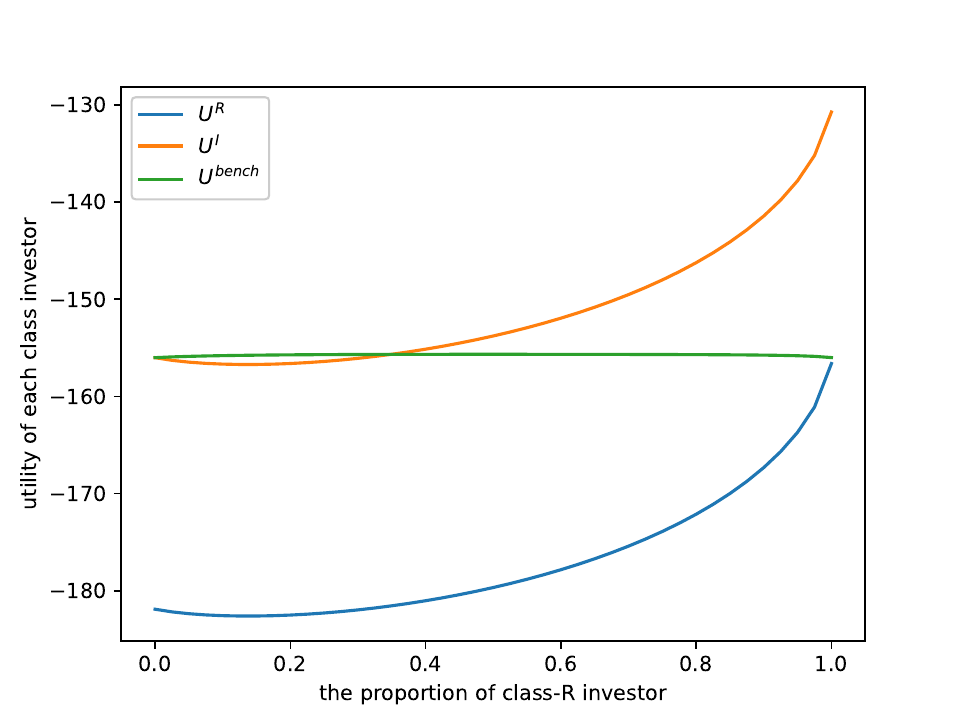}
    \caption{welfare analysis in state altering model.}
 \end{figure}
\end{remark}

Unlike the mean-reverting model, the evolution of $\mu^R_t$ and $\mu^I_t$ follows a nonlinear model, which results in $\frac{d\mu^I_s(\zeta)}{d\zeta}|_{\zeta=0}$ not necessarily being independent of $\mu_s + \mu^R_s-2\mu^I_s(0)$ while they are independent in mean-reverting model. Consequently, there exists a value of $\zeta$ for which $E\big((\mu^R_s-\mu^I_s)^2+(\mu^I_s-\mu_s)^2 -(\mu^R_s-\mu_s)^2 \big) <0$. By observing this, we can provide a counterexample with the following settings: $\mu^h=0.1$, $\mu^l=0.099$, $\lambda=0.2$, $\psi=0.2$, $\mu^R_0 = \mu^I_0 = \mu_0=\mu^l$, and $\rho=100$. Through numerical calculations, we find that $U^{I'}(-0.01)(e^R)|_{e^R=0}>0$. The setting of $\rho=100$ ensures that $\text{sign}(U^{I'}) = \text{sign}(1-E[\frac{\eta^R_{0^+}}{\eta^I_{0^+}}])$, while the setting $\mu^h=0.1$, $\mu^l=0.099$ guarantees that $\text{sign}(E[\frac{\eta^R_{0^+}}{\eta^I_{0^+}}]-1) = \text{sign} \big((\mu^R_{0^+}-\mu^I_{0^+})^2+(\mu^I_{0^+}-\mu_{0^+})^2 -(\mu^R_{0^+}-\mu_{0^+})^2 \big)$, as we can approximate $\exp(x)-1$ with $x$ if $x$ is close to zero. Hence, $U^{I'}(\zeta)(e^R)|_{e^R=0}$ is no longer negative. Nevertheless, this counterexample is deliberately constructed, with more realistic setting that $\mu^h=0.1$, $\mu^l=-0.2$ and $\rho=0.02$, the phenomenon of double loss still happens for any $\zeta$ upon numerical examination. 

At last, the conclusion stated in  \textbf{Theorem \ref{theorem3}} remains valid in the unbiased case for any economic models. Here are the figures about the evolution of the estimations where the background color reflects the actual economic state, where grey refers to the high state and white to the low state and their utility. This is due to the fact that we only rely on the property of conditional expectation just as the proof shown. The phenomenon of double loss still occurs in the unbiased case no matter what economic models we adopt.

\section{Optimal strategy of all-knowing investors}
We now turn our attention to a new problem that builds upon the previous content. In this scenario, there are three types of investors in the market: Class-I investors, Class-R investors, and an all-knowing investor who is aware of the true state of the economy and the degree of signal biases. Passive investors are excluded as they do not affect the market. We assume that the all-knowing investor's strategies are restricted to deciding whether to be a Class-I, Class-R, or passive investor in order to achieve the highest utility. For various reasons, these investors are unable to trade in the market over time due to a lack of market access in certain areas or a lack of attention. Therefore, they delegate their funds to institutions that belong to a certain class.

The answer is straightforward - choose the strategy that maximizes $U^R$, $U^I$, or $U^{bench}$, i.e., $\arg\max\limits_{R,I,ben}(U^R, U^I, U^{bench})$, depending on which investor can achieve the highest utility. We will not solve this numerically, but instead demonstrate it through various settings of the model.

Furthermore, we assume that the all-knowing investors possess an unbiased external signal and can manipulate it to create a biased version into the market. It is reasonable to assume that the degree of bias in the signal can influence the initial proportion of investors in each class, meaning that $e^R$ is a function of $\zeta$. To construct a reasonable function, we observe that if the signal is more biased, fewer investors are likely to trust it. Moreover, if $\zeta$ approaches positive or negative infinity, then no investors will trust the signal. That is to say, $e^{R'}(\zeta)<0$ if $\zeta<0$, $e^{R'}(\zeta)>0$ if $\zeta>0$ and $e^R(\pm \infty) = 1$. 

We define the function set $\mathcal{E}(a,b)$ as follows:
\begin{equation*}
e^R(\zeta) := 1 - a \exp( -b \zeta^2),
\end{equation*}
where $a$ represents the proportion of Class-I investors in the unbiased case, and $b$ represents the rate of decay at which Class-I investors stop trusting the signal and become Class-R investors as the signal becomes more biased. This function satisfies the condition mentioned earlier.

\begin{figure}
  \centering
  \begin{subfigure}[ht]{0.48\textwidth}
    \includegraphics[width=\textwidth]{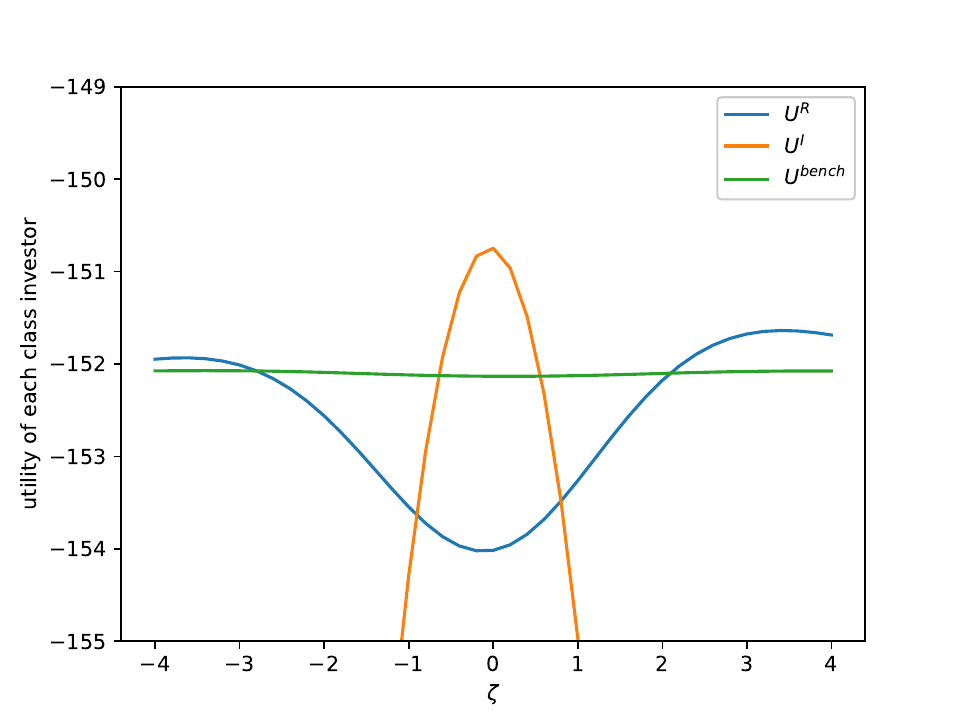}
    \caption{$\mathcal{E}(0.3,0.1)$ .}
  \end{subfigure}
  \hfill
  \begin{subfigure}[ht]{0.48\textwidth}
    \includegraphics[width=\textwidth]{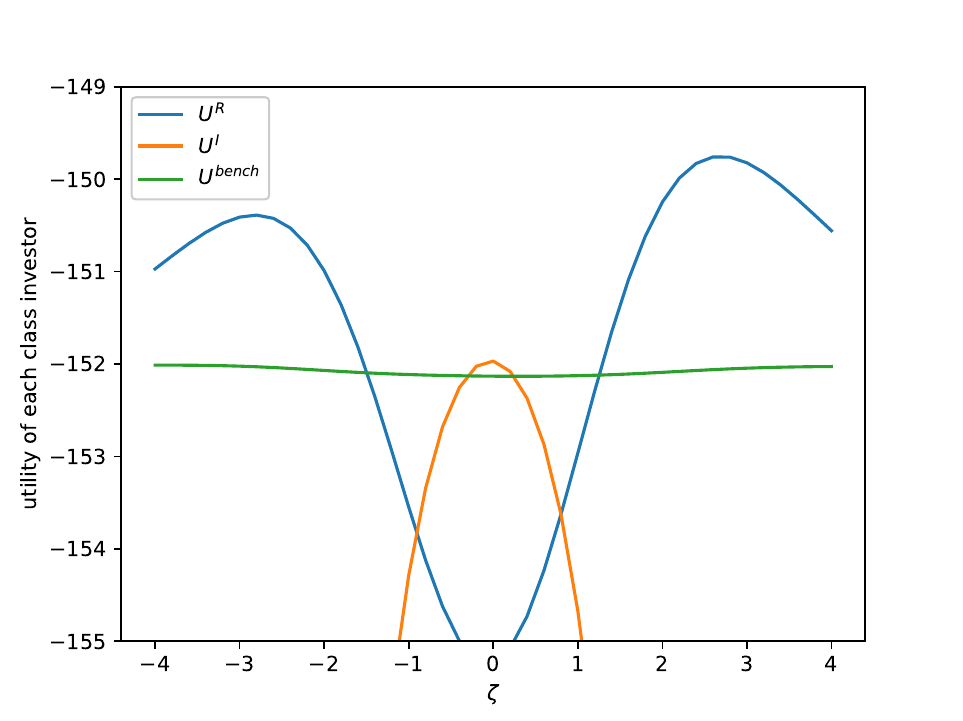}
    \caption{$\mathcal{E}(0.7,0.1)$ .}
    \label{fig:figure2}
  \end{subfigure}
  \centering
  \begin{subfigure}[h]{0.48\textwidth}
    \includegraphics[width=\textwidth]{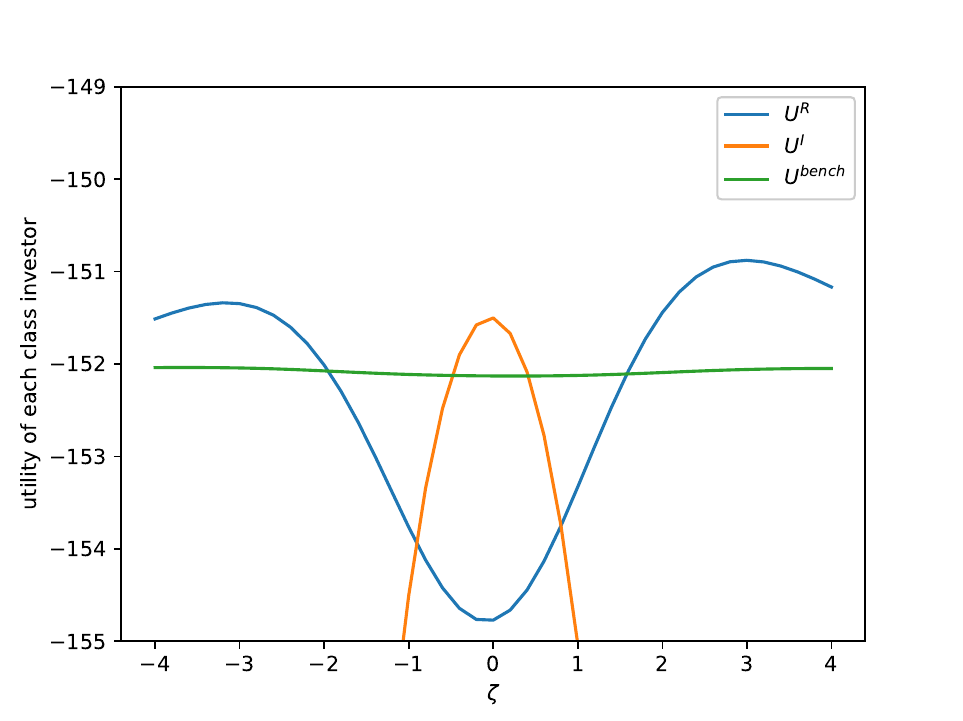}
    \caption{$\mathcal{E}(0.5,0.1)$ .}
  \end{subfigure}
  \hfill
  \begin{subfigure}[h]{0.48\textwidth}
    \includegraphics[width=\textwidth]{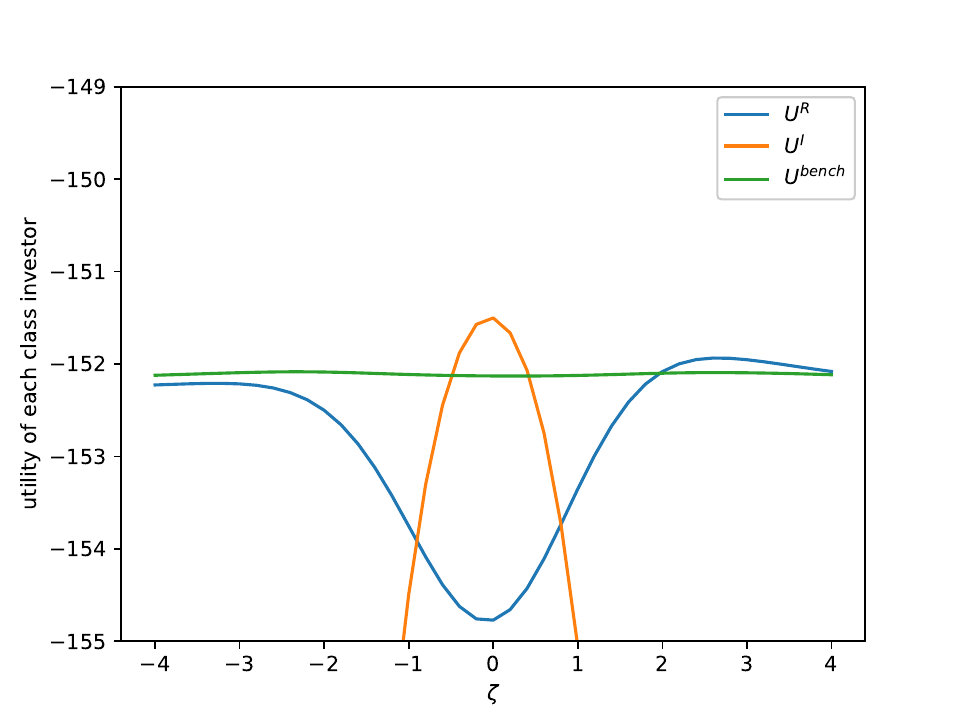}
    \caption{$\mathcal{E}(0.5,0.2)$ .}
    \label{fig:figure2}
  \end{subfigure}
  \centering
  \begin{subfigure}[h]{0.48\textwidth}
    \includegraphics[width=\textwidth]{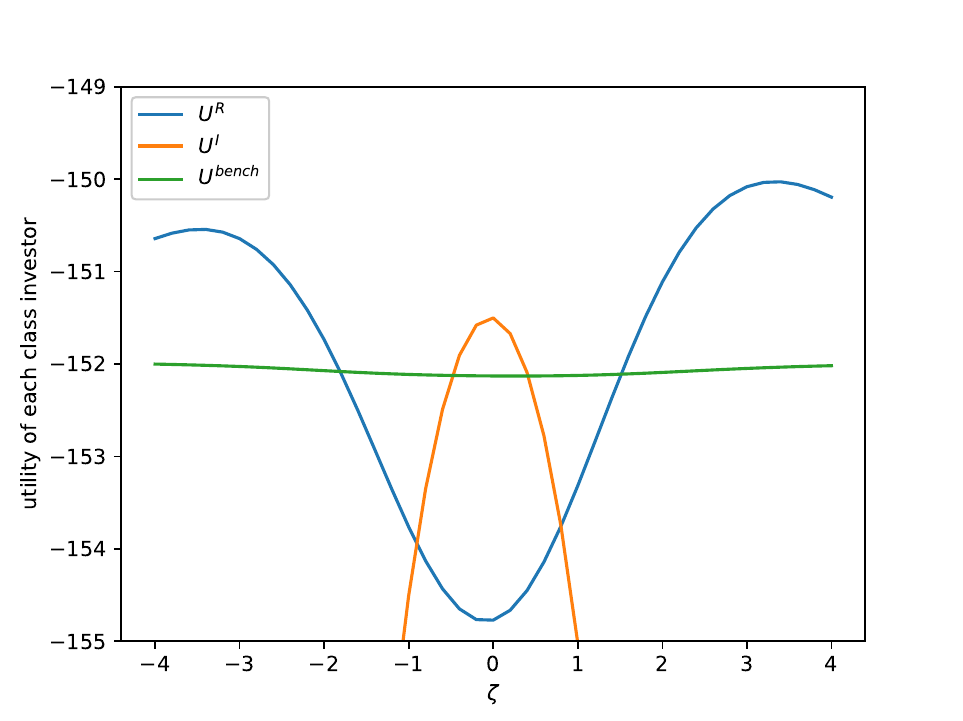}
    \caption{$\mathcal{E}(0.5,0.07)$ .}
  \end{subfigure}
  \hfill
  \begin{subfigure}[h]{0.48\textwidth}
    \includegraphics[width=\textwidth]{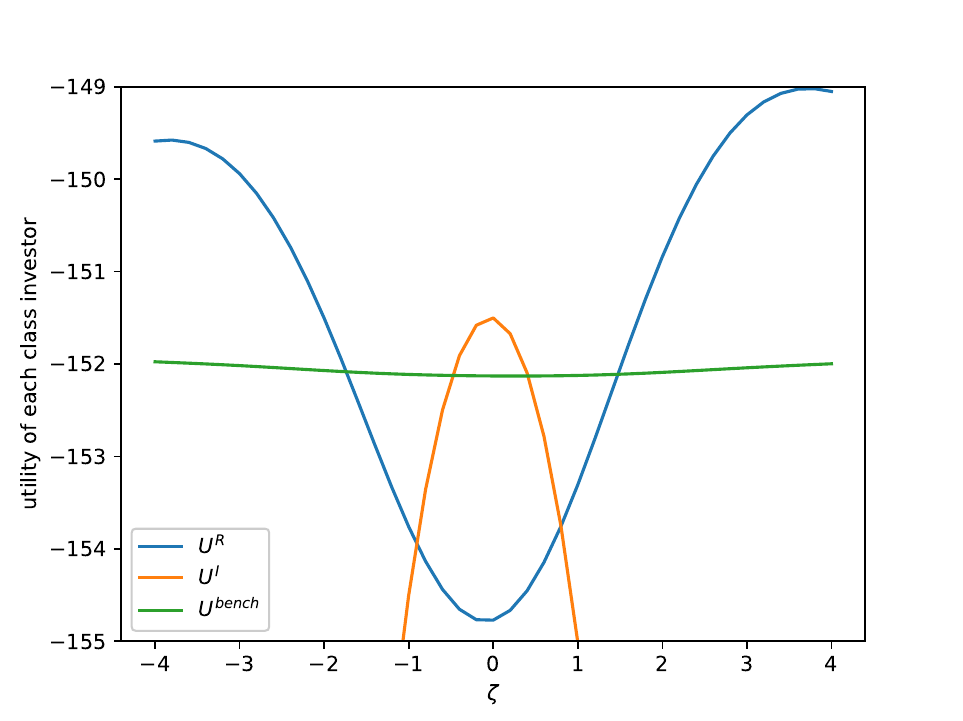}
    \caption{$\mathcal{E}(0.5,0.05)$ .}
    \label{fig:figure2}
  \end{subfigure}
  \caption{The utility of each class investors.}
\end{figure}

The first three figures below illustrate the utility of the investors with varying proportions of Class-I investors in the unbiased case (a). It can be observed that if the proportion of Class-I investors is small, their optimal strategy is to use the information as it is and become Class-I investors. Conversely, if there are plenty of Class-I investors in the biased case, the optimal strategy is to spread the biased information and choose to become Class-R investors.

The last four figures that follow represent the utility of the investors with varying decay rates ($b$). It is evident that if the decay rate is significant, there is no further welfare improvement by manipulating investors' beliefs and becoming Class-R investors. However, if the decay rate is small, there is a significant benefit to manipulating beliefs. Less decay rate, more biased the signal and more welfare improvement.

Overall, there are various methods to enhance utility beyond relying on unbiased signals in the market. One such method is intentionally disseminating misleading information and acting as an antagonist to signal-chasing investors. This approach can result in significant welfare improvements, particularly when the majority of investors are in Class-I.  A more precise statement would be: Sufficient counter-party participation is crucial for all-knowing investors to achieve higher utility.

\section{The survival of each class investors}
In this section, we analyze the asymptotic behavior of the proportion of consumption levels. We have the consumption level defined as
\begin{equation*}
c^m_t = e^{-\rho t} \frac{\eta^m_t}{y^m \xi_t},
\end{equation*}
where $m$ is either I or R. Then we have
\begin{equation*}
\frac{c^I_t}{c^R_t} = k\frac{\eta^I_t}{\eta^R_t} =k \eta_t,
\end{equation*}
where $k$ is a constant. Whether
\begin{equation*}
\lim\limits_{t\to \infty} \frac{c^I_t}{c^R_t} \to 0 \ \text{a.s.} \ \text{or} \ \lim\limits_{t\to \infty} \frac{c^I_t}{c^R_t} \to \infty \ \text{a.s.}
\end{equation*}
depends on the asymptotical property of $\eta_t$. In fact, we have
\begin{equation}\label{logzeta}
\log(\eta_t)(\zeta)=\frac{1}{2}\int_0^t \bigg((\mu^R_s-\mu_s)^2-(\mu^I(\zeta)_s-\mu_s)^2\bigg)h_D^2ds +\int_0^t(\mu^I_s-\mu^R_s)h_DdW_s,
\end{equation}
where $\mu_s$ is the state variable of the economy, and $\mu_s^I(\zeta)$ and $\mu_s^R$ are the estimated state variables for Class-I and Class-R investors, respectively. Then, we have
\begin{equation*}
E\bigg[\log(\eta_t)(\zeta)\bigg]=\frac{1}{2}h_D^2\int_0^t E\bigg[(\mu^R_s-\mu_s)^2-(\mu^I(\zeta)_s-\mu_s)^2\bigg]ds.
\end{equation*}

We claim that 
\begin{lemma}\label{lemma2}
	$(\mu_t, \mu^R_t)$ and $(\mu_t, \mu^I(\zeta)_t)$ have unique stationary distribution.
\end{lemma}

\begin{proof}
From the property of the mean-reverting process. 
\end{proof}

Now we obtain our conclusion:
\begin{theorem}\label{theorem4}
	We can identify constants $\zeta_3<0$ and $\zeta_4 >0$ such that if $\zeta\in (\zeta_3,\zeta_4)$, then $\lim\limits_{t\to \infty}E\bigg[\log(\eta_t)(\zeta)\bigg] \to -\infty$, which implies that $\lim\limits_{t\to \infty} \frac{c^I_t}{c^R_t} \to \infty$ almost surely. Conversely, if $\zeta\in (-\infty, \zeta_3)\cup(\zeta_4, \infty )$, then $\lim\limits_{t\to \infty}E\bigg[\log(\eta_t)(\zeta)\bigg] \to +\infty$, which implies that $\lim\limits_{t\to \infty} \frac{c^I_t}{c^R_t} \to 0$ almost surely.
\end{theorem}

\begin{proof}
See Appendix D.
\end{proof}

By utilizing the expression provided in Definition 1 of \cite{kogan2006price}, this theorem provides insight into the survival of Class-I investors under conditions of slight or heavy bias in the signal. Specifically, in the long run, Class-I investors are said to survive, or experience relative extinction, depending on the degree of signal bias. Conversely, the survival of Class-R investors follows the opposite pattern.

To illustrate this phenomenon, we present simulated trajectories of the consumption proportion. Notably, the case of $\zeta=C$ is similar to $\zeta=-C$, as demonstrated in the figure of utility. For this reason, we only depict the case where $\zeta \geq 0$.

\begin{figure}
  \centering
  \begin{subfigure}[ht]{0.48\textwidth}
    \includegraphics[width=\textwidth]{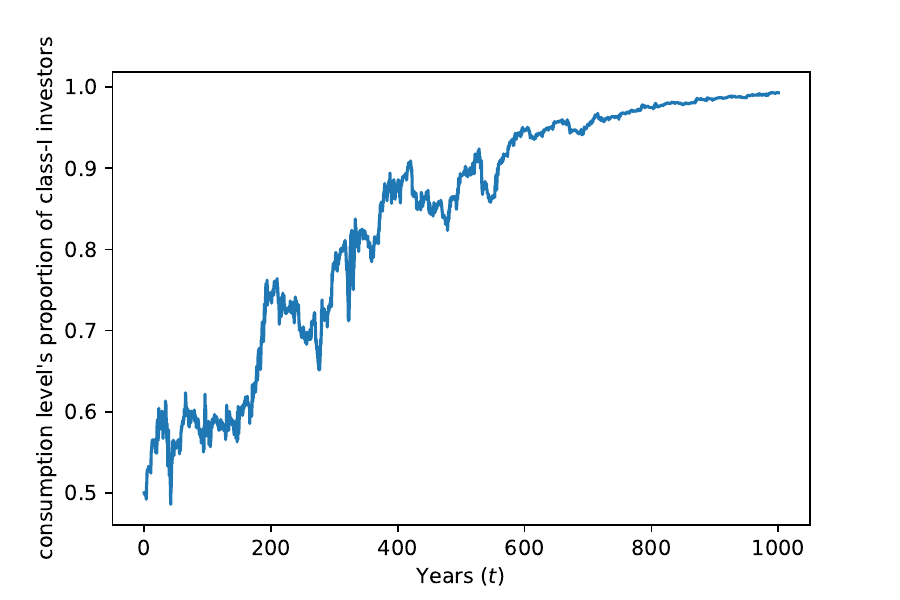}
    \caption{$\zeta=0$ .}
  \end{subfigure}
  \hfill
  \begin{subfigure}[ht]{0.48\textwidth}
    \includegraphics[width=\textwidth]{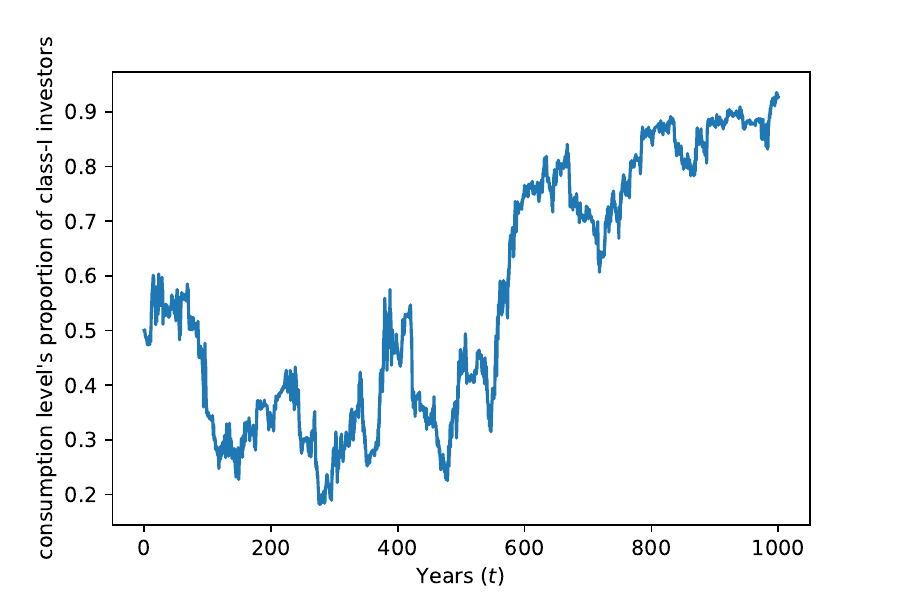}
    \caption{$\zeta=1$ .}
    \label{fig:figure2}
  \end{subfigure}
  \centering
  \begin{subfigure}[h]{0.48\textwidth}
    \includegraphics[width=\textwidth]{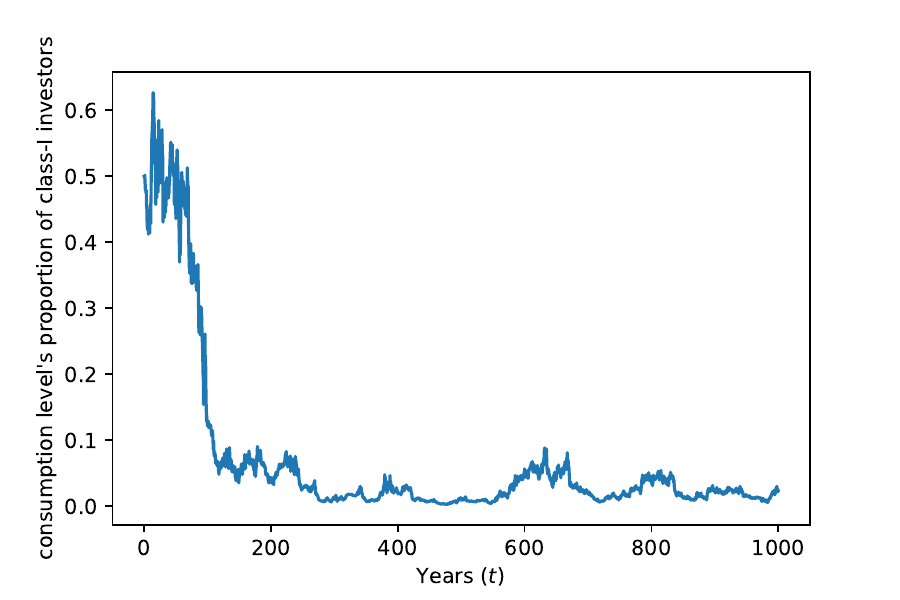}
    \caption{$\zeta=2$ .}
  \end{subfigure}
  \hfill
  \begin{subfigure}[h]{0.48\textwidth}
    \includegraphics[width=\textwidth]{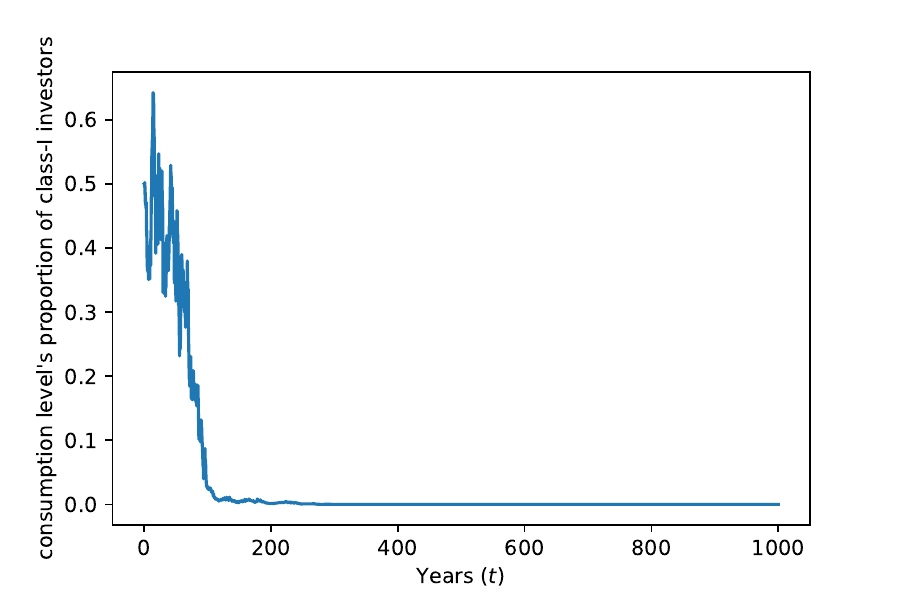}
    \caption{$\zeta=3$ .}
    \label{fig:figure2}
  \end{subfigure}
  \caption{The consumption proportion of Class-I investors.}
\end{figure}

The comparison between $U^I$ and $U^R$ involves a critical point for each utility function, denoted as $\zeta_1$ and $\zeta_2$, respectively. This prompts the question: is there any relationship between these critical points and those of the survival functions, denoted as $\zeta_3$ and $\zeta_4$?

Remarkably, $\zeta_1$, $\zeta_2$, $\zeta_3$, and $\zeta_4$ are equal when the correlation coefficient $\rho$ is zero. However, when $\rho=0$, a potential issue arises because $U^I$ and $U^R$ become infinite and meaningless. To address this, we adopt a trick from infinitely repeated games in game theory, by modifying the utility function as follows: $\lim\limits_{T\to \infty} \frac{1}{T}E[\int_0^T u(t,c_t)dt]$. This modification allows us to interpret the maximization of utility as attempting to survive and consume for as long as possible.

By applying this trick, we find that $\zeta_1$ and $\zeta_2$ are equal to $\zeta_3$ and $\zeta_4$. In other words, the early-time consumption level affects $U^I$ and $U^R$ only if $\rho > 0$.

Another reason for the equivalence between $\zeta_1$, $\zeta_2$, $\zeta_3$, and $\zeta_4$ is that if we assume the initial state $\mu_0$, $\mu^R_0$, and $\mu^I_0$ follow the steady distribution, then we can express $U^I - U^R$ as:
$$
	U^I-U^R =E \bigg[ \int_0^\infty e^{-\rho t}\big( (\mu^R_\infty-\mu_\infty)^2-(\mu^I_\infty(\zeta)-\mu_\infty)^2 \big)t \ dt   \bigg]\\
$$
$$
	= C E\big( (\mu^R_\infty-\mu_\infty)^2-(\mu^I_\infty(\zeta)-\mu_\infty)^2 \big),
$$
where $C>0$ is a constant. As both sides of the equation have the same signal, they are equivalent.

This assumption can be justified by considering that the initial state of the investors is determined by learning from history. Although they trade at time zero, there are plenty of historical statistics to help shape their initial estimation of the economy's state. If the length of the historical period is sufficiently long, we can assume that the initial state follows the steady distribution. 

We do not need to worry about the case where $\zeta=\zeta_3$ and $\zeta=\zeta_4$, in which the limitation of $\lim\limits_{t\to \infty} \frac{c^I_t}{c^R_t}$ cannot be assured. We assume that the circumstances where the degree of biases is equal to a certain special number will never occur.

Furthermore, if there are passive investors in the market, it is evident that they will always survive in the market as their consumption level $c^P_t$ equals $e^P D_t$. They do not trade and only consume the dividends distributed by the risky asset.

\section{Conclusion}

We show that utilizing an unbiased or slightly biased signal for trading can result in higher welfare, as measured by expected realized utility, for investors belonging to Class-I, compared to those who do not use the signal (Class-R). However, if the signal is heavily biased, the situation is reversed.

Furthermore, we introduce passive investors into the market, who do not affect the equilibrium. Their expected utility serves as a useful benchmark for assessing the impact of trading. In the unbiased case, $U^I > U^R$. If the proportion of Class-I investors ($e^R$) is too high, then $U^{bench} > U^I$, indicating that the benefits gained from trading are outweighed by the harm to market efficiency. This phenomenon is also observed in the biased case, implying a double loss for all investors.

We also consider a hypothetical scenario and find that all-knowing investors can improve their welfare not only by utilizing the information advantage of a signal but also by manipulating the other class's estimation of the economy state using a biased signal. This strategy can be particularly effective when the majority of investors are Class-I investors. Sufficient counter-party participation is crucial for all-knowing investors to achieve higher utility.

Finally, we analyze the survival of each type of investor and find that higher welfare and long-term survival usually occur simultaneously if we assume $\rho=0$ or that the initial state follows the steady distribution. In the case of an unbiased or slightly biased signal, Class-I investors are the winners, while Class-R investors win when the signal is heavily biased.

\appendix

\section*{Appendix A. Proof of  \textbf{Theorem \ref{theorem2}}}
We derive this result from $(U^I-U^R)(\zeta)|_{\zeta=0} > 0$ and the following two lemmas.

\begin{lemma}\label{lemma3}
If $\zeta \to +\infty$ or $-\infty$, then $\mu^I_t(\zeta) \to +\infty$ or $-\infty$, respectively. In other words, the Class-I investors will believe the market state always stays in the extremely high or low state. Consequently,
\begin{equation*}
U^I(\zeta) < U^R(\zeta).
\end{equation*}
\end{lemma}

\begin{proof}
The proof is the same as before, with the only difference being that we need to prove the reverse relation with
\begin{equation*}
E[(\mu^R_s-\mu_s)^2] < E[( \pm \infty -\mu_s )^2] = +\infty
\end{equation*}
\end{proof}

Moreover, if we treat the gap $U^I-U^R$ as a function with respect to $\zeta$, we have the following lemma:
\begin{lemma}\label{lemma4}
$(U^I-U^R)(\zeta)$ is an increasing function if $\zeta<0$ and a decreasing function if $\zeta>0$.
\end{lemma}

\begin{proof}
In fact, we observe that $\Delta \mu^I_t(\zeta) :=  \mu^I_t(\zeta)- \mu^I_t(0)$ is a determined process with respect to $\zeta$, (noticing $d\Delta \mu^I_t(\zeta) = -\kappa \Delta \mu^I_t(\zeta)dt -\gamma^I_t (h_D^2+h_e^2) \Delta \mu^I_t(\zeta)dt +\gamma^I_t h_e \zeta dt $) thus it is independent to any stochastic process. 

We calculate
\begin{equation*}
	E[\frac{d \log\eta_t^I(\zeta) }{d\zeta}] =E[ a_t(\zeta)] ,
\end{equation*}
where
\begin{equation*}
	a_t(\zeta) := -\int_0^t \Delta \mu^I_s(\zeta) \frac{d\Delta \mu^I_s(\zeta)}{d\zeta} h_D^2ds,
\end{equation*}
\begin{equation*}
	E[\frac{d^2 \log\eta_t^I(\zeta) }{(d\zeta)^2}] = b_t(\zeta) ,
\end{equation*}
where 
\begin{equation*}
	b_t(\zeta)  := -\int_0^t (\frac{d\Delta \mu^I_s(\zeta)}{d\zeta})^2h_D^2 ds.
\end{equation*}
We find that $E[ a_t(\zeta)] $ and $b_t <0$. Therefore, $\frac{d(U^I-U_R)(\zeta)}{d\zeta}|_{\zeta=0}=0$ and $\frac{d^2(U^I-U_R)(\zeta)}{(d\zeta)^2} <0$. 
\end{proof}

\section*{Appendix B. Proof of \textbf{Lemma \ref{lemma1}}}
As for the benchmark, we have:
\begin{equation*}
D_t=c^R_t+c^I_t=e^{-\rho t}\frac{1}{\xi_t} \left(\frac{\eta^R_t}{y^R}+ \frac{\eta^I_t}{y^I}\right),
\end{equation*}
\begin{equation*}
P_0=X^R_0+X^I_0=\frac{1}{\rho}\left( \frac{1}{ y^R}+\frac{1}{ y^I}\right).
\end{equation*}

Given the pricing kernel, we can determine the risky asset's price through its dividend as follows:
\begin{equation*}
\xi_t P_t = E\left[\int_t^\infty \xi_s D_s | \mathcal{F}_t\right].
\end{equation*}
Therefore,
\begin{equation*}
P_t = D_t E\left[\int_t^\infty \frac{\xi_s}{\xi_t}\frac{D_s}{D_t} ds | \mathcal{F}_t\right].
\end{equation*}
If we rewrite $D_t$ as:
\begin{equation*}
D_t = e^{-\rho t}\frac{1}{\xi_t}\left(\frac{\eta^I_t}{y^I}+ \frac{\eta^R_t}{y^R}\right) = e^{-\rho t}\frac{1}{\xi_t}\eta^{bench}_t\left(\frac{1}{y^I}+\frac{1}{y^R}\right),
\end{equation*}
where
\begin{equation*}
\eta^{bench}_t:=\frac{y^I}{y^R+y^I}\eta^R_t+\frac{y^R}{y^R+y^I}\eta^I_t.
\end{equation*}
Substituting it into the expression for $P_t$ yields:
\begin{equation*}
P_t = D_t E\left[\int_t^\infty e^{-\rho (s-t)} \frac{\eta^{bench}_s}{\eta^{bench}_t} ds | \mathcal{F}_t\right].
\end{equation*}

In fact, $\eta^{bench}$ is an exponential martingale.

As we have $\frac{y^I}{y^R}=:k$, then $\eta^{bench}_t$ can be represented as $ \frac{k}{1+k}\eta^R+\frac{1}{1+k}\eta^I$. We have
\begin{align*}
d\eta^{bench}_t &= \frac{k}{1+k}d\eta^R_t+\frac{1}{1+k}d\eta^I_t\\
&= \frac{k}{1+k}\eta^R_t(\mu^R_t-\mu_t)h_DdW_t+\frac{1}{1+k}\eta^I_t(\mu^I_t-\mu_t)h_DdW_t\\
&= \frac{k\eta_t}{1+k\eta_t}\eta^{bench}_t(\mu^R_t-\mu_t)h_DdW_t+\frac{1}{1+k\eta_t}\eta^{bench}_t(\mu^I_t-\mu_t)h_DdW_t\\
&= \eta^{bench}_t(\lambda_t(\mu^R_t-\mu_t)+(1-\lambda_t)(\mu^I_t-\mu_t))h_DdW_t
\end{align*}
where the third equation holds for the relation $\eta_t=\frac{\xi^I_t}{\xi^R_t}=\frac{\xi_t/\eta^I_t}{\xi_t/\eta^R_t}=\frac{\eta^R_t}{\eta^I_t}$, and we calculate that $\eta_t^R= \frac{(1+k)\eta_t}{1+k\eta_t}\eta^{bench}_t$ and $\eta^I_t= \frac{1+k}{1+k\eta_t}\eta^{bench}_t$.

Thus $E\big[\frac{\eta^{bench}_s}{\eta^{bench}_t}\big] =1$ and
\begin{equation*}
P_t = \frac{1}{\rho} D_t.
\end{equation*}
It implies that the price-earning ratio is a constant.

Furthermore, we see that $P_0 = \frac{1}{\rho} D_0$. This corresponds to the fact that $U^{bench}=E\bigg[ \int_0^\infty u\big(t,\frac{\rho D_t}{D_0}\big) dt \bigg]$ is a constant.

\section*{Appendix C. Proof of  \textbf{Theorem \ref{theorem3}}}
We define a functional as follows:
\begin{equation*}
U(\eta) = E\left[ \int_0^\infty e^{-\rho t} \log\left(e^{-\rho t}\frac{\rho \eta_t}{\xi_t}\right) dt \right],
\end{equation*}
where $\eta$ is a stochastic process. Then we see $U^I = U(\eta^I)$ and $U^R = U(\eta^R)$. Moreover, we have
\begin{equation*}
U^{bench} = E\left[ \int_0^\infty u(t,\frac{D_t}{P_0}) dt \right] = E\left[ \int_0^\infty e^{-\rho t} \log\left(e^{-\rho t}\frac{\rho \eta^{bench}_t}{\xi_t}\right) dt \right].
\end{equation*}
Thus we have
\begin{equation*}
U^{bench} = U(\eta^{bench}).
\end{equation*}
As for the initial proportion $e^R$, we have $\frac{e^R}{e^I}=\frac{X_0^R}{X_0^I}=\frac{y^I}{y^R}=k$. With $e^I+e^R=1$, we have $e^R=\frac{y^I}{y^R+y^I}$, and we can rewrite $\eta^{bench}_t$ as $\eta^{bench}_t = e^R \eta^R_t + e^I \eta^I_t$.

Now, let's consider $U^I(e^R)$ and make some techniques to analyze it. We consider the difference $U^I(e^R)-U^{bench}(e^R)$. Similarly, 
\begin{equation*}
U^I - U^{bench} = E\left[ \int_0^\infty e^{-\rho t} \left(\log(\eta_t^I)-\log(\eta_t^{bench})\right) dt \right].
\end{equation*}
This can help us eliminate the effect of changes in $\xi_t$.

We calculate that:
\begin{equation*}
\frac{d}{de^R}(U^I-U^{bench})(e^R)=-E\left[ \int_0^\infty e^{-\rho t} \frac{\eta^R_t-\eta^I_t}{e^R\eta^R_t+(1-e^R)\eta^I_t} dt \right],
\end{equation*}
\begin{equation*}
\frac{d^2}{d(e^R)^2}(U^I-U^{bench})(e^R)=E\left[ \int_0^\infty e^{-\rho t} \frac{(\eta^R_t-\eta^I_t)^2}{(e^R\eta^R_t+(1-e^R)\eta^I_t)^2} dt \right]\geq 0.
\end{equation*}

Taking $e^R=0$ and $e^R=1$ into the above equations, we have:
\begin{equation*}
\frac{d}{de^R}(U^I-U^{bench})(e^R)\bigg|_{e^R=0}=-\int_0^\infty e^{-\rho t} \left(E\left[\frac{\eta^R_t}{\eta^I_t}\right]-1\right)dt,
\end{equation*}
\begin{equation*}
\frac{d}{de^R}(U^I-U^{bench})(e^R)\bigg|_{e^R=1}=-\int_0^\infty e^{-\rho t} \left(1-E\left[\frac{\eta^I_t}{\eta^R_t}\right]\right)dt.
\end{equation*}

Using the equation:
\begin{equation*}
\eta^m_t = \exp\left(-\frac{1}{2}\int_0^t (\mu^m_s-\mu_s)^2h_D^2ds +\int_0^t(\mu^m_s-\mu_s)h_DdW_s \right).
\end{equation*}

We simplify the expressions as follows:

\begin{equation*}
E\left[\frac{\eta^R_t}{\eta^I_t}\right]= E\left[\exp\left(\frac{1}{2}\int_0^t \left((\mu^R_s-\mu^I_s)^2+(\mu^I_s-\mu_s)^2 -(\mu^R_s-\mu_s)^2 \right)h_D^2ds \right)\right],
\end{equation*}

\begin{equation*}
E\left[\frac{\eta^I_t}{\eta^R_t}\right]= E\left[\exp\left(\frac{1}{2}\int_0^t \left((\mu^R_s-\mu^I_s)^2+(\mu^R_s-\mu_s)^2 -(\mu^I_s-\mu_s)^2 \right)h_D^2ds \right)\right].
\end{equation*}

Given that $\mu^R_s=E[\mu_s|\mathcal{F}^D_s]$, $\mu^I_s=E[\mu_s|\mathcal{F}^D_s \vee\mathcal{F}^e_s]$, and $\mu^R_s=E[\mu^I_s|\mathcal{F}^e_s]$, we can use the property of conditional expectation to obtain:
\begin{equation*}
E\bigg[(\mu^R_s-\mu^I_s)^2+(\mu^I_s-\mu_s)^2 -(\mu^R_s-\mu_s)^2 \bigg]=0.
\end{equation*}

Applying Jensen's inequality gives
\begin{equation*}
E\left[\frac{\eta^R_t}{\eta^I_t}\right]> \exp\bigg( E\bigg[\frac{1}{2}\int_0^t \bigg((\mu^R_s-\mu^I_s)^2+(\mu^I_s-\mu_s)^2 -(\mu^R_s-\mu_s)^2 \bigg)h_D^2ds \bigg] \bigg)=1.
\end{equation*}

Thus, we have:
\begin{equation*}
\frac{d}{de^R}(U^I-U^{bench})(e^R)|_{e^R=0} < 0.
\end{equation*}

We also calculate:
\begin{equation*}
E\bigg[(\mu^R_s-\mu^I_s)^2+(\mu^R_s-\mu_s)^2 -(\mu^I_s-\mu_s)^2 \bigg]>0.
\end{equation*}

Using this inequality, we get:
\begin{equation*}
E\left[\frac{\eta^I_t}{\eta^R_t}\right]\geq \exp\bigg( E\bigg[\frac{1}{2}\int_0^t \bigg((\mu^R_s-\mu^I_s)^2+(\mu^R_s-\mu_s)^2 -(\mu^I_s-\mu_s)^2 \bigg)h_D^2ds \bigg] \bigg)>1
\end{equation*}

and
\begin{equation*}
\frac{d}{de^R}(U^I-U^{bench})(e^R)|_{e^R=1}> 0.
\end{equation*}

As $U^{bench}$ is constant, $U^I$ can inherit all the properties of $U^I-U^{bench}$ stated above. Noting that the gap between $U^I$ and $U^R$ is constant, we can conclude that $U^R$ inherits all the properties as well.

Regarding the biased cases, we only need to modify two relations that
\begin{align*}
	&E\bigg[(\mu^R_s-\mu^I_s(\zeta))^2+(\mu^I_s(\zeta)-\mu_s)^2 -(\mu^R_s-\mu_s)^2 \bigg] \\
	=&E\bigg[(\mu^R_s-\mu^I_s(0))^2+(\mu^I_s(0)-\mu_s)^2 -(\mu^R_s-\mu_s)^2 \bigg] + 2 E(\Delta \mu^I_t(\zeta))^2> 0
\end{align*}
\begin{align*}
	&E\bigg[(\mu^R_s-\mu^I_s(\zeta))^2+(\mu^R_s-\mu_s)^2 -(\mu^I_s(\zeta)-\mu_s)^2 \bigg] \\
	=&E\bigg[(\mu^R_s-\mu^I_s(0))^2+(\mu^R_s-\mu_s)^2 -(\mu^I_s(0)-\mu_s)^2 \bigg] + 0 E(\Delta \mu^I_t(\zeta))^2>0
\end{align*}
and the other content of proof remains same.

\section*{Appendix D. Proof of  \textbf{Theorem \ref{theorem4}}}
\textbf{Lemma \ref{lemma2}} shows that $\lim\limits_{s\to \infty}E\bigg[(\mu^R_s-\mu_s)^2\bigg]$ and $\lim\limits_{s\to \infty}E\bigg[(\mu^I_s(\zeta)-\mu_s)^2\bigg]$ exist, which we denote them as $E\bigg[(\mu^R_\infty-\mu_\infty)^2\bigg]$ and $E\bigg[(\mu^I_\infty(\zeta)-\mu_\infty)^2\bigg]$ for convenience. The proof in \textbf{Lemma \ref{lemma4}} also shows that $E(\mu^I_t(\zeta)-\mu_t)^2$ for any $t$ is an increasing function if $\zeta>0$ and a decreasing function if $\zeta<0$. From the analysis of the limit, we conclude that $E\bigg[(\mu^I_\infty(\zeta)-\mu_\infty)^2\bigg]$ inherit the same property with respect to $\zeta$.

	Similarly, we have $E\bigg[(\mu^R_\infty-\mu_\infty)^2-(\mu^I_\infty(0)-\mu_\infty)^2\bigg]>0$ and $E\bigg[(\mu^R_\infty-\mu_\infty)^2-(\mu^I_\infty(\pm \infty)-\mu_\infty)^2\bigg]<0$. Therefore, there exist constants $\zeta_3<0$ and $\zeta_4>0$ such that if $\zeta\in(\zeta_3,\zeta_4)$, then $E\bigg[(\mu^R_\infty-\mu_\infty)^2-(\mu^I_\infty(\zeta)-\mu_\infty)^2\bigg]>0$, and if $\zeta\in(-\infty,\zeta_3)\cup(\zeta_4,\infty)$, then $E\bigg[(\mu^R_\infty-\mu_\infty)^2-(\mu^I_\infty(\zeta)-\mu_\infty)^2\bigg]<0$. Consequently, using (\ref{logzeta}), we have $\lim\limits_{t\to \infty}E\bigg[\log(\eta_t)(\zeta)\bigg] \to +\infty$ for $\zeta\in(-\infty,\zeta_3)\cup(\zeta_4,\infty)$ and $\lim\limits_{t\to \infty}E\bigg[\log(\eta_t)(\zeta)\bigg] \to -\infty$ for $\zeta\in(\zeta_3,\zeta_4)$.

\end{document}